\documentclass[11pt, titlepage]{article}
\usepackage[left=2.2cm,right=2.2cm,top=2.5cm,bottom=2.5cm]{geometry} 
\usepackage{amsfonts,amssymb,amsthm,amsmath,amscd}
\usepackage{enumerate}
\usepackage{tikz}

\usepackage{graphicx}

\usepackage{hyperref}
\usepackage[capitalise,noabbrev]{cleveref} 
\crefname{assumption}{assumption}{assumptions}
\crefname{auxassumption}{assumption}{assumptions}
\crefname{appassumption}{condition}{conditions}

\usepackage{overpic} 
\usepackage{url} 
\usepackage{comment}

\usepackage{graphicx} 
\usepackage{amsmath,amsthm,amssymb,color,latexsym}
\usepackage{mathtools}
\usepackage{mathrsfs}
\usepackage{bbm}
\usepackage[shortlabels]{enumitem}
\usepackage{mleftright}
\usepackage{microtype}
\usepackage[inkscapelatex=false]{svg}
\usepackage[many]{tcolorbox}
\usepackage{setspace}  
\usepackage{multicol}
\usepackage{tabularx}
\usepackage{scalerel}
\usepackage[authoryear, square]{natbib} 
\usepackage{xparse}
\usepackage{algorithm}
\usepackage{algpseudocode}
\usepackage{aliascnt}

\theoremstyle{theorem}

\newtheoremstyle{break}
  {3pt}{3pt}    
  {}          
  {}            
  {\bfseries}  
  {.}          
  {\newline}     
  {}

\theoremstyle{break}
\newtheorem{problem}{Problem}

\theoremstyle{definition}
\newtheorem{theorem}{Theorem}[section]

\newaliascnt{lemma}{theorem}
\newtheorem{lemma}[lemma]{Lemma}
\aliascntresetthe{lemma}

\newaliascnt{corollary}{theorem}
\newtheorem{corollary}[corollary]{Corollary}
\aliascntresetthe{corollary}

\newaliascnt{proposition}{theorem}
\newtheorem{proposition}[proposition]{Proposition}
\aliascntresetthe{proposition}

\newaliascnt{definition}{theorem}
\newtheorem{definition}[definition]{Definition}
\aliascntresetthe{definition}

\newaliascnt{remark}{theorem}

\aliascntresetthe{remark}

\newaliascnt{example}{theorem}

\aliascntresetthe{example}

\newtheorem{assumption}{Assumption}

\newtheorem{auxassumption}{Assumption}

\newtheorem{appassumption}{Assumption}

\crefname{appsec}{Appendix}{Appendices}
\Crefname{appsec}{Appendix}{Appendices}
\crefname{appsubsec}{Appendix}{Appendices}
\Crefname{appsubsec}{Appendix}{Appendices}

\newcommand\restr[2]{{  \left.\kern-\nulldelimiterspace  #1  \right|_{#2} }}

\newcommand\diff{\mathop{}\!\mathrm{d}}

\newcommand\dP{\diff \kern 0.033em \mathbb{P}}

\newcommand\cev[2]{\mathbb{E}\mleft[ #1 \mleft\vert #2 \mright.\kern-\nulldelimiterspace \mright]}

\newcommand\cprob[2]{\mathbb{P}\mleft( #1 \mleft\vert #2 \mright.\kern-\nulldelimiterspace \mright)}
\newcommand{\indep}{\perp \!\hspace{-0.3em}\! \perp}

\NewDocumentCommand{\dodistr}{m m m}{
  p\!\left(
  #1
  \if\relax\detokenize{#2}\relax
    \,\middle\|\,
  \else
    \mid #2 \,\middle\|\,
  \fi
  #3
  \right)
}

\newcommand{\dsep}[1][\mathcal{G}]{\perp_{#1}}

\newcommand{\pa}[1][\mathcal{G}]{\textnormal{pa}_{#1}}
\newcommand{\an}[1][\mathcal{G}]{\textnormal{an}_{#1}}
\newcommand{\de}[1][\mathcal{G}]{\textnormal{de}_{#1}}
\newcommand{\dis}[1][\mathcal{G}]{\textnormal{dis}_{#1}}

\newcommand{\mb}[1][\mathcal{G}]{\textnormal{mb}_{#1}}

\newcommand{\setRV}[1]{\mathbf{#1}}

\tcbset{
    sharp corners,
    colback = white,
    before skip = 8pt,    
    after skip = 8pt      
}                           

\newtcolorbox{boxA}{
    boxrule = 1.5pt,
    colframe = black 
}

\newtcolorbox{kernelbox}[1]{
    colback=white,
    colframe=black,
    title={#1},
    fonttitle=\bfseries,
    boxrule=0.8pt,
    arc=2pt,
    left=6pt,
    right=6pt,
    top=6pt,
    bottom=6pt
}

\renewenvironment{abstract}
  {\vspace{1.2em}
   \noindent\textbf{Abstract.}\ }
  {\par\vspace{0.9em}}

\begin{document}

\thispagestyle{empty}

\begin{center}
{\LARGE\bfseries Identifying Interventional Joint Distributions\\
via Extended Bridge Functions}

\vspace{0.8em}

{\large Constantin Schott}

\vspace{0.25em}

{\small University of Cambridge}

\vspace{0.15em}

{\small \texttt{cts37@cam.ac.uk}}

\vspace{0.5em}

{\small \today}
\end{center}

\begin{abstract}
Existing identification results in proximal causal inference often focus on marginal
interventional distributions using standard outcome- or treatment bridge functions. These methods do not generally identify joint interventional distributions that contain all proxy variables that were used to define the corresponding bridge functions. In many applications, however, these joint interventional distributions are a natural target of interest. We introduce extended bridge functions and derive new identification results for joint interventional distributions that may retain all relevant proxy variables. We then apply these results to proximal identification algorithms, where interventional kernels naturally arise as intermediate objects, yielding a generalized framework based on kernel operations.
\end{abstract}


\vspace{1.2em}

\section{Introduction}
\label{sec:Introduction}
Causal inference from observational data is becoming increasingly important in many disciplines, as it seeks to bridge the gap between association and causation. Association addresses questions of the form ``What do I expect to happen if I observe \ldots?'' or ``How likely is event $A$, given that event $B$ has been observed?''. Formally, associational knowledge is encoded in the joint probability distribution of a collection of observed random variables, which we refer to as the \textit{observational distribution}.  As a running example, we consider a medical study in which a treatment variable $A$ (influenza vaccination) may influence an outcome variable $Y$ (infectious respiratory-disease-related hospitalization), along with measured covariates $X$ (age, gender). In this context, the observational distribution can be estimated using an observational study in which vaccination status $A$, covariates $X$, and hospitalization outcomes $Y$ are recorded for a cohort of patients. Crucially, the investigator does not intervene in the population dynamics but merely measures variables as they naturally occur. In applied settings, such observational data are often readily available and comparatively inexpensive to obtain \citep{hernan2025}.

In many settings, researchers need to evaluate the effect of actively performing an intervention, trying to answer questions of the form ``What happens if I do \ldots?''. Formally, this is represented by a distinct post-intervention joint distribution, $(Y(a), X(a))$, over the observed variables, induced by the (hypothetical) intervention that sets $A = a$. We refer to this as an \textit{interventional distribution}. In the medical trial, such interventional distributions can be empirically approximated in randomized controlled trials by randomly assigning each participant to a treatment $A = a$ and subsequently measuring the resulting outcome $Y$ and covariates $X$. In contrast to observational studies, interventional data are hard to obtain in practice due to experimental, logistical, technical, or ethical constraints \citep{hernan2025}. 

In many practical situations, researchers need to evaluate causal effects but are unable to conduct a randomized trial. Instead, they are confined to working with purely observational data. In our example medical trial, randomizing patients to receive a vaccine is not feasible. On the other hand, it is easily possible to ask patients about their vaccination status and hospitalization events. Consequently, our task is to use the readily accessible observational distribution to identify the hard-to-obtain interventional distribution. The main challenge lies in bias due to hidden confounders. These are unobserved variables that influence both treatment and outcome, thereby inducing bias that cannot be removed by standard statistical adjustment methods.

We begin by providing a tutorial on causal identification theory and present recently developed \emph{proximal causal inference} methods that show how, under certain conditions, identification of the interventional distribution remains possible in the presence of hidden confounding. Proximal causal inference uses proxy variables of the hidden confounders. Such proxy variables are often abundantly available in many observational studies \citep{TchetgenTchetgen2024}. We present two different identification strategies based on proxy variables: identification via \emph{outcome bridge functions} \citep{Miao2018} and identification via \emph{treatment bridge functions} \citep{Cui2023}. Outcome bridge functions have been extensively applied by \citet{Shpitser2021} to the general identification problem. We derive analogous results for treatment bridge functions in \cref{sec:outcome_and_treatment_bridge_functions}. Our primary contribution is the introduction of new \emph{extended bridge functions} that allow for the identification of the entire interventional joint distribution in \cref{sec:extended_bridge_functions}. Finally, in \cref{sec:proximal_identification_theory}, we combine outcome, treatment, and extended bridge functions with classical identification theory and propose a framework for generalized proximal identification algorithms as an extension of the results in \citet{Shpitser2021}.

\section{Notation and conventions}
\label{sec:notation}
For a random variable $Y$, let $\mathcal{X}_Y$ denote its state space, $P_Y
$ denote its distribution, and $p(y)$ denote the probability function for discrete $Y$ or the probability density function for continuous $Y$. We always silently assume that all random variables have a joint density with respect to an underlying dominating product measure and that all conditional probabilities and densities are well-defined. When manipulating densities and probability functions, we use a summation sign that is to be understood as integration in the continuous case. Additionally, we drop the domain we summing over (e.g. $W \in \mathcal X_W$) if it is clear from the context. Any equation between densities is implicitly understood to hold ``only up to the uniqueness of the density'', i.e., except for a null set of the corresponding dominating measure. For discrete random variables, we use the notation from \citet{Miao2018} and let capitalization denote matrix notation, i.e., $p(U \, |\, Z,a,x) = (p(u_i \, | \, z_j,a,x))_{i,j=1}^{n,m}$ if $\mathcal{X}_U = \{u_1, \dots, u_n\}$ and $\mathcal{X}_Z = \{z_1, \dots, z_m\}$.
Furthermore, we denote d-separation in a graph $\mathcal{G}$ by $\dsep$ and conditional independence by $\indep$. We denote the power set by $\mathcal{P}$. 

\section{Causal models and the identification problem}
\label{sec:Causal_models_and_the_identification_problem}
To relate observational and interventional distributions, we need to specify a causal model that provides a mechanistic description of the relationship between observed variables. We adopt the widely used \textit{nonparametric structural equation model with independent errors} (NPSEM-IE) from \citet{Pearl2009} as our causal framework. Many less restrictive models are discussed in the literature and could be used instead of the NPSEM-IE, for example, the FFRCISTG model from \citet{Robins1986}. Much of causal theory is formulated in terms of graphs with vertices corresponding to random variables. We refer to \citet{Pearl2009} for the basic graph-theoretic definitions and recall the associated notations in \cref{sec:Graphs_and_DAGs}. 
\begin{definition}[NPSEM-IE] 
    \label{def:NPSEM-IE}
    Consider a set of random variables $V :=\{V_1, \dots, V_n\}$. Then $V$ follows an NPSEM-IE model with error variables $\epsilon_1, \dots, \epsilon_n$ taking values in $\mathcal{X}_{\epsilon_1}, \dots, \mathcal{X}_{\epsilon_n}$ and sets of parent variables $\pa[](V_i) \subseteq V$ if, for all $i = 1, \dots, n$:
    \begin{align*}
        V_i = f_i(\pa[](V_i), \epsilon_i),
    \end{align*}
    where the noise variables $\epsilon_1, \dots, \epsilon_n$ are jointly independent and each $f_i: \mathcal{X}_{\pa[](V_i)} \times \mathcal{X}_{\epsilon_i} \rightarrow \mathcal{X}_{V_i}$ is a measurable function 
    , called a \textit{structural equation}. We require that the graph $\mathcal{G}$ induced by $\pa(V_i) = \pa[](V_i)$ is a DAG, called the \textit{causal graph} of the NPSEM-IE. We call $p(V)$ the \textit{observational distribution}. 
    We summarize the model as a tuple $(V, f,\epsilon)$ with $f = (f_1, \dots,f_n)$ and $\epsilon = (\epsilon_1, \dots, \epsilon_n)$. For simplicity, we refer to the NPSEM-IE $(V,f, \epsilon)$ simply as a \textit{causal model}.
    
    For $A \subset V$ and $a \in \mathcal{X}_A$, we define the intervention operation $\text{do}(A=a)$ mapping $(V, f,\epsilon)$ to $(\tilde{V} , \tilde{f}, \epsilon)$ by replacing the structural equations for $A$ with the constant assignment $A = a$ and keeping the structural equations for $Y := V \setminus A$ unchanged. This yields a modified model $(\tilde{V}, \tilde{f}, \epsilon)$ defined by 
    \begin{alignat*}{2}
        &\tilde{A}_j \equiv a_j \quad &&\text{for} \quad A_j \in A, \\
        &\tilde{V}_i = f_i(\tilde{\text{pa}}(\tilde{V}_i), \epsilon_i) \quad &&\text{for} \quad V_i \in Y,
    \end{alignat*}
    such that the resulting causal graph $\tilde{\mathcal{G}}$ is a \textit{mutilated graph} of $\mathcal{G}$ \citep{Pearl2009}, keeping all edges from $\mathcal{G}$ except those leading to a variable $\tilde{A}_j$ that has been intervened on. We refer to the random variables $V_i(a) := \tilde{V}_i$ as \textit{potential outcomes} and define the distributions of the form $p(Y(a))$ for some $Y \subseteq V \setminus A$ as the \textit{interventional distributions} of $V_i$ after the intervention $\text{do}(A=a)$.
\end{definition}

\begin{figure}
        \centering
        \includegraphics[width=\linewidth]{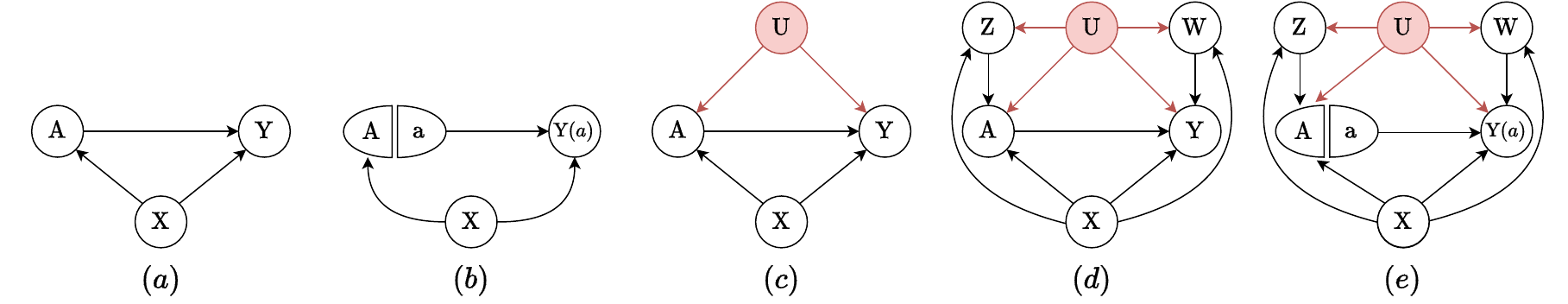}
        \caption{\textbf{(a)} A causal graph depicting the causal relations between $A$ (influenza vaccination), $Y$ (respiratory-disease-related hospitalization) and $X$ (age, gender) in the example medical trial.  \textbf{(b)} The corresponding SWIG for the causal model in $(a)$ after the intervention $\text{do}(A=a)$. 
        \textbf{(c)} Causal graph (a) with an added unmeasured confounder $U$ (health awareness) influencing both  treatment $A$ and outcome $Y$.
        \textbf{(d)} Causal graph (c) with treatment-inducing proxy $Z$ (cancer screening) and outcome-inducing proxy $W$ (traumatic injuries).
        \textbf{(e)} The corresponding SWIG for the causal model in $(d)$ after $\text{do}(A=a)$.}
        \label{fig:standard_proximal_graph}
\end{figure}

In the following, we will always assume that we are working within an NPSEM-IE causal model. The definition of the causal model introduces interventional distributions. These are examples of (Markov) kernels \citep{Shpitser2021}, which we also write as $\dodistr{X,Y}{}{a} := p(X(a),Y(a))$ with density $\dodistr{x,y}{}{a}$ and conditional density $\dodistr{x}{y}{a}$ of $X(a)$ given $Y(a) =y$. We refer to \citet{Kallenberg2021} for a rigorous introduction to (Markov) kernels and a proof of the existence of conditional densities for kernels. When referring to the entire kernel without specifying an interventional value assignment, we also write $\dodistr{X}{Y}{A}$ and refer to $A$ as the intervention set and $Y$ as the conditioning set. 

In the medical trial example discussed in \cref{sec:Introduction}, we take \cref{fig:standard_proximal_graph}(a) as a plausible causal graph. In the corresponding NPSEM-IE model, the structural equations can be interpreted as causal mechanisms that assign $A$ its observed value based on $X$ and $Y$ its observed value based on $X$ and $A$. For the intervention $\text{do}(A=a)$, we set $A = a$ constant while keeping the mechanisms for $X$ and $Y$ unchanged. We will now list some direct consequences of the NPSEM-IE model. Importantly, we can conclude conditional independence statements between factual variables $V_i$ and even between certain factual and counterfactual variables $V_i(a)$ from graphical criteria on causal graphs or subsidiary graphs derived from the causal graph. The most important graphical criterion is d-separation (see \citet{Pearl2009}) applied to the causal graph directly or to a Single World Intervention Graph (SWIG). In the following, SWIGs will  be used only as graphical tools to deduce ignorability conditions from causal graphs. We do not provide a formal introduction to SWIGs here and instead refer to \citet{Richardson2013}.
\begin{lemma}[d-separation in causal graphs and SWIGs] 
    In the NPSEM-IE, d-separation on the causal graph implies conditional independence of the corresponding random variables. For all disjoint $X,Y,Z \subseteq V$ in a causal graph $\mathcal{G}$
    \begin{align}
        \label{eq:d_sep_in_NPSEM}
        X \dsep Y \, |\, Z \quad \Longrightarrow \quad X \indep Y \,|\, Z.
    \end{align}
    We say that $V$ satisfies the \emph{global Markov property} with respect to the graph $\mathcal{G}$.
    Fix an intervention $\text{do}(A = a)$. If we apply the \textit{node splitting operation} described in \citet{Richardson2013} to a causal graph, we can apply d-separation to the resulting \textit{Single World Intervention Graph} (SWIG) $\mathcal{G}(A=a)$. Then, the set of random variables $A \cup (V \setminus A)(a)$ satisfies the global Markov property with respect to the SWIG.
    \footnote{To apply the theory of d-separation and factorization to a SWIG, we formally need to define these concepts in the case of mixed graphs containing both random and non-random variables. We refer to \citet{Richardson2013} or  \citet{Zhao2022} for a discussion of this issue.}
    \begin{proof}
        We refer to Chapter 3 in \citet{Pearl2009} for a proof of \eqref{eq:d_sep_in_NPSEM} and to \citet{Richardson2013} for a proof of the global Markov property for SWIGs.
    \end{proof}
\end{lemma}

The SWIG of our example trial under the intervention $\text{do}(A=a)$ is shown in \mbox{\cref{fig:standard_proximal_graph}(b)} and allows us to conclude the \textit{no unmeasured confounder} condition $Y(a) \indep A \, | \, X$, which is also called the \emph{ignorability} condition. Thus, in our model, the choice of vaccine $A$ is independent of the potential outcomes $Y(a)$ for any $a \in \mathcal{X}_A$, once we adjust for the covariates $X$. In this sense, we call $X$ a \textit{sufficient adjustment set}. We can use this conditional independence to calculate the interventional distribution $\dodistr{Y}{}{a}$. In general, conditional independence relations that we can deduce from the causal graph lead to a factorization of the observational and interventional distributions.
\begin{lemma}[Factorization and g-computation] 
    In the setting of \cref{def:NPSEM-IE}, the observational distribution $p(V)$ factorizes over the parent sets as
    \begin{align}
        p(V) = \prod_{V_i \in V} p(V_i \, | \, \pa(V_i)).
    \end{align}
    For every $A \subseteq V$ and $a \in \mathcal{X}_A$, the interventional distribution $\dodistr{Y}{}{a}$ of $Y = V \setminus A$ after $\text{do}(A= a)$ decomposes into the truncated factorization 
    \begin{align}
        \label{eq:truncated_factorization}
        \dodistr{Y}{}{a}= \prod_{V_i \in Y} p(V_i \, | \, \pa(V_i)) \big|_{A = a},
    \end{align}
    where $p(V_i \, | \, \pa(V_i)) \big|_{A = a}$ signifies that we substitute the corresponding value $a_j$ for any $A_j \in \pa(V_i) \cap A$ in the corresponding factor $p(V_i \, | \, \pa(V_i))$. This factorization is commonly referred to as the \textit{g-formula.}
    \begin{proof}
        We refer to Chapter 3 in \citet{Pearl2009} for a proof and discussion of these results.
    \end{proof}
\end{lemma}
Notably, the g-formula only uses conditionals of the observational distribution and the parent structure of the causal model encoded in its causal graph. Thus, if we have a causal model with its causal graph and its entire joint observational distribution, we can identify all interventional distributions! In our example, this yields the standard backdoor adjustment formula after marginalization. If we further assume positivity $p(a \, | \, x) > 0$ for all $(a,x) \in \mathcal{X}_{(A,X)}$, the corresponding inverse probability weighted (IPW) formula follows directly from the definition of conditional densities.
\begin{alignat}{4}
    &\text{(Backdoor)}\qquad \qquad \qquad \dodistr{y}{}{a} &\; = \;&\sum_{x} p(y \mid a,x)\,p(x)
    &\qquad& \label{eq:backdoor_formula}\\
    &\text{(IPW)}\qquad \qquad \qquad                &\;=\; &\sum_{x} \frac{p(y,a,x)}{p(a \mid x)}.
    &\qquad& \label{eq:ipw_formula}
\end{alignat}
Alternatively, both the backdoor and the IPW formula can be deduced directly from algebraic conditions. As we will use versions of this argument repeatedly, we introduce consistency, ignorability, and positivity conditions and give a detailed proof for future reference. 
\begin{proposition}
    \label{prop:basic_backdoor_IPW_formalae}
    Under the following assumptions, the interventional kernel $\dodistr{y}{}{a}$ can be identified using the backdoor formula \eqref{eq:backdoor_formula} or the IPW formula \eqref{eq:ipw_formula}.
    \begin{enumerate}
        \item \textbf{Consistency:} $Y(A) = Y$ almost surely,
        \item \textbf{Positivity:} for all $a \in \mathcal{X}_A$: $p(a \, | \, X) > 0$  $P_X$-almost surely,
        \item \textbf{Ignorability:} $Y(a) \indep A \, |\, X$ for all $a \in \mathcal{X}_A$.
    \end{enumerate}
    \begin{proof}
        Consistency implies that for any choices of $y,a$ and $x$, the sets $\{Y(a) =y, A=a, X=x\}$ and $\{Y=y, A=a,X=x\}$ differ only on a null set. As we assume the existence of joint and conditional densities, we have  $p(Y(a) = y \,|\,  A=a, X=x) = p(Y = y \,|\,  A=a, X=x)$, $P_{(A,X)}$-almost surely. The positivity assumption then allows us to strengthen this result to $(P_A \otimes P_X)$-almost sure equality. Then we easily calculate
        \begin{align*}
            \dodistr{y}{}{a} &= \sum_{x} p(Y(a) = y \,|\,  X = x) \, p(X=x) \\
            &= \sum_{x} p(Y(a) = y \,|\,  A=a, X = x) \, p(X=x) &\text{(by ignorability)}\\
            &= \sum_{x} p(Y = y \,|\,  A=a, X = x) \, p(X=x) &\text{(by consistency and positivity)} \\
            &= \sum_{x} \frac{p(Y = y ,  A=a, X = x)}{p(A=a \, | \, X=x)} . &\text{(by positivity)}
        \end{align*}
    \end{proof}
\end{proposition}

We note that consistency is implicit in the structural equations of the NPSEM-IE framework and will not be stated as a separate assumption in the following. Positivity is, in many settings, also a relatively mild assumption \citep{hernan2025}. In contrast,  ignorability is often unrealistic in practice, as it requires that all relevant confounders are measured and  included in $X$. In the example trial described in \cref{sec:Introduction}, the health awareness $U$ of a given patient could act as a confounder, influencing both the treatment $A$ (vaccination) and the outcome $Y$ (hospitalization). However, it is very difficult to quantify ``health awareness'' in a methodologically sound way, effectively rendering $U$ unobserved (also called latent). Therefore, we need to consider partially unobserved causal models over a set of random variables $V \cup U$, where only the joint distribution $p(V)$ over $V$ is observed and $U$ is latent. The corresponding causal graph is shown in \mbox{\cref{fig:standard_proximal_graph}(c)}. Can we still determine the effect that intervening on $A$ will have on $Y$? This naturally leads to the nonparametric identification problem that has been widely discussed in the literature. We refer to \citet{Shpitser2021} for a recent review.

\begin{boxA}
\begin{problem}[Nonparametric identification problem] 
\label{prob:nonparametric_identification_problem}
Consider a causal model with observed variables $V$ and latent variables $U$.
Assume that the model induces the causal graph $\mathcal{G}(V \cup U)$.
Fix an intervention set $A \subseteq V$ and an outcome set $Y \subseteq V \setminus A$.
Under these assumptions, are the interventional distributions $\dodistr{Y}{}{a}$
uniquely determined by the observational distribution $p(V)$?
\end{problem}
\end{boxA}

Since we are interested in \textit{unique} identification, we need to consider the set of all causal models that agree on a fixed observational distribution $p(V)$ and induce the same causal graph $\mathcal{G}(V \cup U)$. Our objective is to determine whether an interventional distribution $\dodistr{Y}{}{a}$ is identical across all these models. Then we can derive identification formulas that uniquely recover these interventional kernels $\dodistr{Y}{}{a}$ from $p(V)$. There are multiple approaches to solving this problem, including backdoor adjustment and front-door adjustment \citep{Pearl2009}, culminating in the ID algorithm in \citet{Shpitser2006} that provides a sound and complete identification algorithm. We discuss the ID algorithm in \cref{sec:Identification_algorithms}. 

Coming back to our example setting, we can construct two NPSEM-IE models with causal graph \mbox{\cref{fig:standard_proximal_graph}(c)} that generate the same distribution on the observed variables $A,Y$ and $X$, but different interventional distributions $\dodistr{Y}{}{a}$ \citep{Pearl2009}. This shows that $\dodistr{Y}{}{a}$ is \textit{not identified} in our example because of the unobserved confounding variable $U$! 

To make progress in our example, we need to further restrict the set of causal models that we are considering by imposing additional assumptions. This restricted set of models can be expected to agree on more interventional distributions. In other words, under stronger assumptions, we can obtain correspondingly stronger identification results, potentially enabling the identification of $\dodistr{Y}{}{a}$. 

\section{Outcome and treatment bridge functions}
\label{sec:outcome_and_treatment_bridge_functions}
\subsection{Motivation: Proximal causal Inference}
The main ideas of proximal causal inference (PCI) originate from the concept of \textit{negative control variables}, which have been widely used in experimental biology and epidemiology \citep{Lipsitch2010}. In our example, we introduce two additional observed variables $Z$ and $W$, as illustrated in the causal graph \mbox{\cref{fig:standard_proximal_graph}(d)}. These variables serve as proxies for the unobserved confounder 
$U$.

First, we additionally measure the occurrence of traumatic injuries as a \textit{negative control outcome} $W$.  Based on medical domain knowledge, it is reasonable to assume that $A$, the influenza vaccination, does not causally affect $W$. Furthermore, we can assume that the health awareness $U$ of patients impacts the rate at which traumatic injuries occur. In this sense, we call $W$ \textit{U-relevant}.

Second, we also measure the number of cancer screenings that a patient received as a \textit{negative control exposure} $Z$. Again, domain knowledge suggests that $Z$ has no direct causal effect on either $Y$ or $W$, since a cancer screening should neither cause nor be caused by infectious diseases or traumatic injuries. Similarly to $W$, $Z$ is also influenced by the  health awareness $U$, and is therefore also U-relevant. The corresponding causal graph including $W$ and $Z$ is shown in \mbox{\cref{fig:standard_proximal_graph}(d)}. For completeness, we include possible relationships between $Z$ and $A$ as well as $W$ and $Y$ in the graph. 

We carefully chose $W$ and $Z$ to be $U$-relevant, meaning that $W$ and $Z$ represent observable effects of the unmeasured confounder $U$. Compared to our discussion of causal graphs in \cref{sec:Causal_models_and_the_identification_problem}, we use the dependence of $W$ and $Z$ on $U$, instead of only using (conditional) independencies between variables. Our domain knowledge suggests that after conditioning on $A$, $Y$ and $X$, all association between $W$ and $Z$ that is present in our data arises through their shared dependence on $U$. This suggests that variation in $W$ and $Z$ can be used to indirectly adjust for $U$ when trying to approximate the causal effect that $A$ has on $Y$. In this sense, we can consider $W$ and $Z$ as specially chosen proxies of the latent variable $U$. To emphasize this, we adopt the terminology from \citet{TchetgenTchetgen2024} and will refer to $W$ as an \emph{outcome-inducing proxy} and to $Z$ as a \emph{treatment-inducing proxy}. 

To translate this intuition into a mathematical idea, we recall that both the backdoor adjustment \cref{eq:backdoor_formula} and the IPW formula \cref{eq:ipw_formula} rely on conditioning on a sufficient adjustment set. In our setting, the set $\{U,X\}$ would be sufficient; however, the variable $U$ is unobserved and cannot be directly used for adjustment. Only conditioning on $X$ will, in general, yield biased estimates of $\dodistr{Y}{}{a}$. Proximal causal inference circumvents this issue by using the proxy variables $Z$ and $W$ to modify these identification strategies and to account for unmeasured confounding. First, we take (generally insufficient) adjustment sets $\{W,X\}$ for the backdoor formula and $\{Z,X\}$ for the IPW formula. In order to account for the unmeasured confounder $U$, we modify the term $p(y \, |\, a,x ,w)$ in the backdoor formula using an \textit{outcome bridge function}, or the term $p(a \, |\, x,z)^{-1}$ in the IPW formula using a \textit{treatment bridge function}. Instead of attempting to recover $U$, we only infer how to correctly adjust the relevant components of the identification formulas. 

In the remainder of this section, we first present a method to adjust the backdoor formula using an \textit{outcome bridge function} as developed by \citet{Miao2018, TchetgenTchetgen2024}. This allows for the identification of $\dodistr{Y}{}{a}$ and, under further assumptions, identification of $\dodistr{Y,Z,X}{}{a}$, as detailed in \citet{Shpitser2021}. Second, we discuss the adjustment of the IPW formula using a \textit{treatment bridge function}. Treatment bridge functions were first introduced in \citet{Cui2023} to identify the \textit{average treatment effect} (ATE). We expand on this to identify the distribution $\dodistr{Y}{}{a}$ and, under further assumptions, $\dodistr{Y,W,X}{}{a}$. Finally, we introduce \emph{extended bridge functions} and present two novel identification results for the entire interventional joint distribution $\dodistr{Y,W,Z,X}{}{a}$. To our knowledge, extended bridge functions and identification results for the entire interventional joint distribution have not been presented in the literature before.

\subsection{Setup and assumptions}
An overview of the different assumptions used in the following derivations is shown in \cref{fig:PCI_assumption_overview}. Again, we consider a set of random variables $A,Y,X,W,Z$ and an unobserved variable $U$. As described above, standard identification theory uses conditional independence relations between factual and counterfactual quantities and exclusion restrictions. Within the NPSEM-IE framework, these relations are fully determined by the causal graph or by SWIGs derived from it. In this sense, standard identification theory can be understood as operating purely on the causal graph.\footnote{ Technically, positivity assumptions could be regarded as ``non-graphical'' assumptions required in identification theory, but we ignore this distinction here to emphasize the fundamentally different nature of completeness conditions in PCI.} Similarly, PCI also relies on a set of assumptions that can be represented graphically. To be as general as possible, we present these assumptions algebraically instead of referring to a specific graph.
\begin{assumption}[Latent ignorability]
\label{ass:latent_ignorability}
 $Y(a) \indep A \, | \, U,X$ for all $a \in \mathcal{X}_A$.    
\end{assumption}
\begin{assumption}[Valid outcome-inducing proxy]
\label{ass:valid_outcome_inducing_proxy}
 $W \indep (Z,A) \, | \, U,X$.    
\end{assumption}
\begin{assumption}[Valid treatment-inducing proxy]
\label{ass:valid_treatment_inducing_proxy}
 $Z \indep Y \, | \, A,U,X.$    
\end{assumption}
\begin{assumption}[Latent positivity]
\label{ass:latent_positivity}
For all $a \in \mathcal{X}_A$: $p(a \, |\, U,X) > 0$ $P_{(U,X)}$-almost surely.   
\end{assumption}
Following the standard rules of conditional independence \citep{Pearl2009}, we often also use the direct consequences of \cref{ass:valid_outcome_inducing_proxy}
\begin{alignat*}{4}
    W &\indep Z \, | \, U,X  \quad &\text{and} \quad &W \indep A \, | \, U,X \quad \qquad &&\text{by decomposition,} \\
    W &\indep Z \, | \, A,U,X  \quad &\text{and} \quad &W \indep A \, | \, Z,U,X \quad \qquad &&\text{by weak union}.
\end{alignat*}
The \cref{ass:latent_ignorability,ass:valid_outcome_inducing_proxy,ass:valid_treatment_inducing_proxy} hold, for example, in the causal graph depicted in \mbox{\cref{fig:standard_proximal_graph}(d)} and its SWIG \mbox{\cref{fig:standard_proximal_graph}(e)} after the intervention $\text{do}(A=a)$. We refer to the Appendix in \citet{TchetgenTchetgen2024} for an overview of more possible causal graphs for which \cref{ass:latent_ignorability,ass:valid_outcome_inducing_proxy,ass:valid_treatment_inducing_proxy} hold.

Notably, unlike standard identification theory, PCI additionally requires assumptions that cannot be deduced from the causal graph alone. When we motivated the use of negative control variables, we needed to argue about U-relevance of the proxies $W$ and $Z$. As an extension of U-relevance, PCI needs assumptions to ensure sufficient dependence between $U$ and $W$ as well as $U$ and $Z$. This is formalized through completeness conditions and conditions on the solvability of integral equations. These assumptions are inherently non-graphical and are a key difference between PCI and classical identification theory.

 \begin{figure}[h!]
        \centering
        \includegraphics[width=0.9\linewidth]{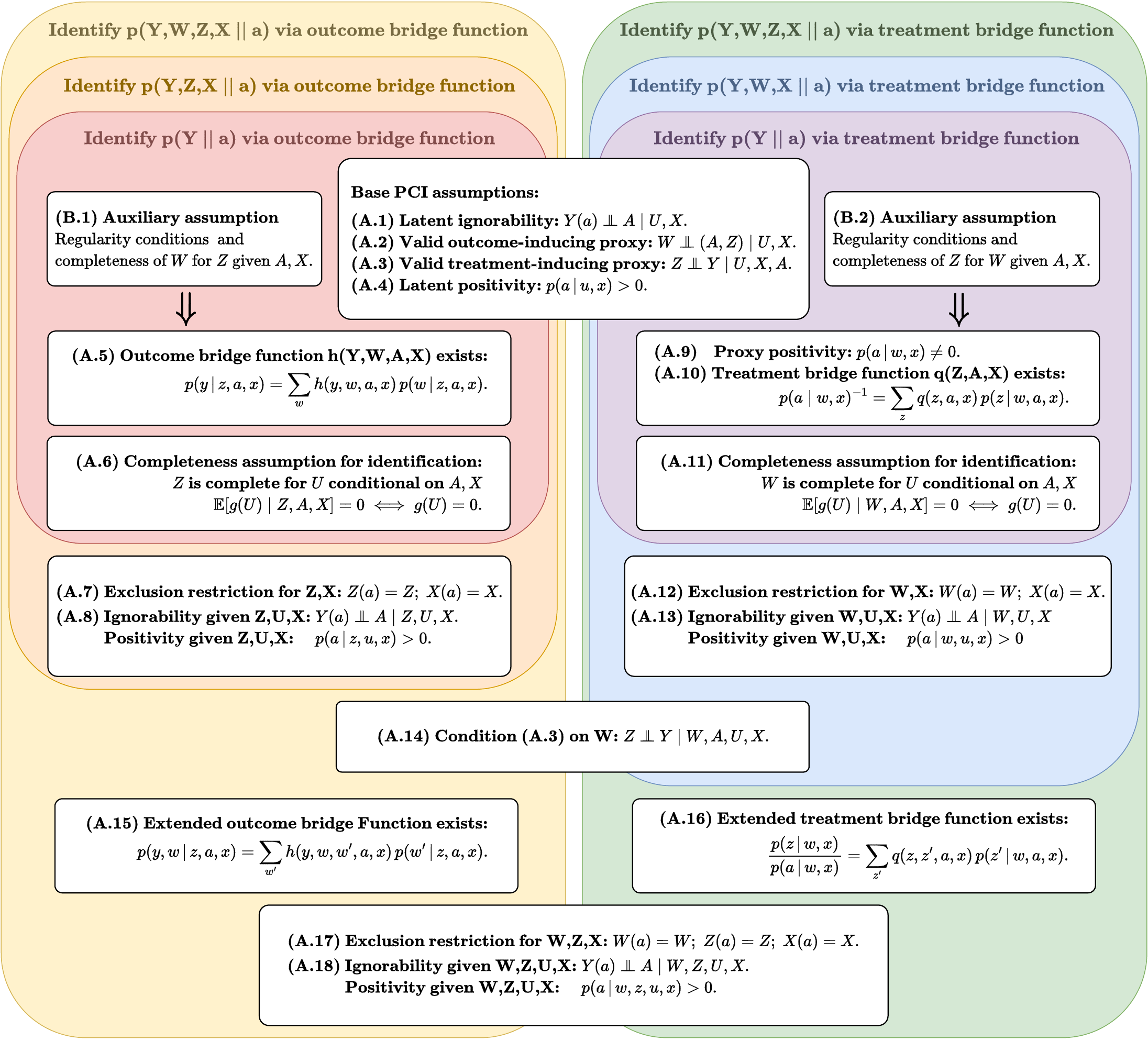}
        \caption{This figure gives an overview of different assumptions used in the PCI identification results in \cref{sec:outcome_and_treatment_bridge_functions}.
        Some of the assumptions were shortened for better readability. The layout of this figure was adapted from \citet{Ringlein2025}.}
        \label{fig:PCI_assumption_overview}
\end{figure}

\subsection{Outcome bridge functions}
\label{sec:Outcome_bridge_functions}
We start by deducing an identification result using an outcome bridge function as a modification of the backdoor formula \cref{eq:backdoor_formula}. The core assumption is the existence of an outcome bridge function $h$.

\begin{assumption}[Existence of outcome bridge function]
    \label{ass:outcome_bridge_function}
    There exists an outcome bridge function $h(y,w,a,x)$, solving the integral equation
    \begin{align}
        \label{eq:outcome_bridge_function}
        \sum_{w} h(y,w,a,x) \, p(w \, | \, z,a,x) = p(y \, | \, z,a,x).
    \end{align}
\end{assumption}
Formally, equation \eqref{eq:outcome_bridge_function} constitutes a \textit{Fredholm equation} of the first kind \citep{Miao2018, Carrasco2007}, which can be regarded as a generalization of a matrix equation. Intuitively, the bridge function $h$ defines a linear integral transformation that recovers the effect of $Z$ on $Y$ from the effect of $Z$ on $W$, both conditional on $A$ and $X$. A sufficient condition for the existence of a solution is given in \cref{sec:conditions_for_existence_of_bridge_functions} in \cref{auxass:outcome_bridge}. Earlier work of \citet{Miao2018, TchetgenTchetgen2024} uses a different bridge equation that targets the conditional expectation $\mathbb{E}[Y \, | \, Z=z, A =a, X =x]$ on the RHS and uses a bridge function $h$ with no dependence on $y$. While \citet{Miao2018} aim to identify the moments $\mathbb{E}[Y(a)]$, we follow the approach discussed in \citet{Shpitser2021} to identify the entire distribution $\dodistr{Y}{}{a}$, which uses the bridge equation in \eqref{eq:outcome_bridge_function}. Notably, the solution $h$ is not necessarily unique. Approaches to approximate $h$ from data include standard Generalized Method of Moments (GMM) methods using parametric function classes for $h$ \citep{TchetgenTchetgen2024,Miao2018}, kernel-based methods \citep{Mastouri2021}, or minimax learning approaches \citep{Kallus2021}.

As a second core assumption, we require completeness of $Z$ for $U$, conditional on $A$ and $X$.
\begin{assumption}[Completeness]
\label{ass:outcome_completeness}
For any $g \in L^2(U)$ and all $a \in \mathcal{X}_A,x \in \mathcal{X}_X$
\begin{align*}
    \mathbb{E}[g(U) \, | \, Z,A=a,X=x] = 0 \text{ almost surely } \ \Longleftrightarrow \ g(U) = 0 \text{ almost surely}.
\end{align*}
\end{assumption}
We note that completeness gives an equivalence between an $P_Z$-almost sure statement and a $P_U$-almost sure statement, for all values of $a \in \mathcal{X}_A$ and $x \in \mathcal{X}_X$. Intuitively, \cref{ass:outcome_completeness} ensures that $Z$ has sufficient variability with respect to $U$ and is U-relevant. If we take $Z$ and $U$ to have finite state spaces, \citet{Miao2018} showed that the completeness assumption reduces to a matrix rank condition.
\begin{corollary}
    \label{cor:discrete_completeness_matrix_rank}
    If $U$ and $Z$ are discrete random variables with state spaces $\mathcal{X}_U = \{u_1, \dots, u_n\}$ and $\mathcal{X}_Z = \{z_1, \dots, z_m\}$, the completeness \cref{ass:outcome_completeness} is equivalent to 
    \begin{align*}
        \forall \; a \in \mathcal{X}_A, x \in \mathcal{X}_X:
        \; (p(u_i \, | \, z_j,a,x))_{i,j=1}^{n,m} \in \mathbb{R}^{n\times m}\;\text{ having full row-rank.}
    \end{align*}
    In particular, this requires $m\geq n$.
    \begin{proof}
        See \citet{Miao2018}.
    \end{proof}
\end{corollary}

With these assumptions we can prove the first proximal identification result, adapting the proof of Theorem 1 in \citet{Shpitser2021}, which is based on the main result in \citet{Miao2018}.

\begin{proposition}[Proximal backdoor formula]
    \label{prop:Proximal_backdoor_formula}
    Under \cref{ass:latent_ignorability,ass:valid_outcome_inducing_proxy,ass:valid_treatment_inducing_proxy,ass:latent_positivity,ass:outcome_bridge_function,ass:outcome_completeness} we can identify the distribution $\dodistr{Y}{}{a}$ as
    \begin{align}
        \label{eq:outcome_bridge_identifying_functional}
        \dodistr{y}{}{a} = \sum_{w,x} h(y,w,a,x)\, p(w,x).
    \end{align}
    
\begin{proof}
We express $p(y \, |\, z,a,x)$ in two different ways. First, by marginalizing over $p(y,u \, | \, a,z,x)$ and second, using the bridge equation. This yields the equation chain
\begin{align*}
\sum_u p(y \, | \, a,u,x) \, p(u \, | \, z,a,x) &= p(y \, |\, z,a,x)  && \text{(by \ref{ass:valid_treatment_inducing_proxy})} \\
&=  \sum_w h(y,w,a,x)\, p(w \,|\, z,a,x)
&& \text{(by \ref{ass:outcome_bridge_function})} \\
&= \sum_u \sum_w h(y,w,a,x)\, p(w \,|\, u,x) \,p(u \, | \, z, a,x).
&& \text{(by \ref{ass:valid_outcome_inducing_proxy})}
\end{align*}
We explicitly state the steps of the completeness argument once in this proof. In all following proofs, we will shorten this argument and appeal to completeness directly. Subtracting the first and last terms in the equation chain yields
\begin{align*}
    0 = \sum_u \underbrace{\left(p(y \, | \, a,u,x) 
    - \sum_{w} h(y,w,a,x) \, p(w \, | \, u,x) \right)}_{=: \, g_{y,a,x}(u)} \, p(u \,| \, z,a,x).
\end{align*}
Under sufficient integrability\footnote{Technically, we would have to add integrability as an additional assumption on $h$. For readability and in accordance with \citet{Shpitser2021} and \citet{TchetgenTchetgen2024}, we always silently assume that bridge functions are sufficiently integrable or even bounded.},
we can apply the completeness assumption \ref{ass:outcome_completeness} to all $g_{y,a,x} \in L^2(U)$ to conclude
\begin{align}
    \label{eq:outcome_bridge_with_u}
    p(y \, | \, a,u,x) = \sum_{w} h(y,w,a,x) \, p(w \, | \, u,x) \quad \text{for all $y \in \mathcal{X}_Y, a \in \mathcal{X}_A, x \in \mathcal{X}_X$, $P_U$-a.s.}
\end{align}
It only remains to apply this identity to the backdoor formula with the adjustment set $\{U,X\}$, which is sufficient by \ref{ass:latent_ignorability} and \ref{ass:latent_positivity}.\footnote{In the NPSEM-IE setting we are considering, consistency holds by definition and we need not add it as an extra assumption.} This yields
\begin{align*}
    \dodistr{y}{}{a} = \sum_{x,u} p(y \, | \, a,u,x) \, p(u,x) = \sum_x \sum_w h(y,w,a,x) \, \sum_u p(w \, | \, u,x) \, p(u,x) = \sum_{w,x} h(y,w,a,x)\, p(w,x).
\end{align*}
\end{proof}
\end{proposition}
Notably, the resulting identifying functional \eqref{eq:outcome_bridge_identifying_functional} has a similar form to the backdoor formula \eqref{eq:backdoor_formula} with insufficient adjustment set $(W,X)$ but exchanging $p(y \, |\, w,a,x)$ for the bridge function $h(y,w,a,x)$. In the next step, we similarly extend this approach to identify the joint distribution $\dodistr{Y,Z,X}{}{a}$ with a modification of the truncated factorization \eqref{eq:truncated_factorization} using the outcome bridge function $h$. The result is a special case of Theorem 5 in \citet{Shpitser2021}, which shows a similar result for general probability kernels. We require the following stronger assumption. By inspecting the causal graph \mbox{\cref{fig:standard_proximal_graph}(d)} and SWIG \mbox{\cref{fig:standard_proximal_graph}(e)}, we conclude that these stronger assumptions hold in our example setting. 
\begin{assumption}[Exclusion restrictions for $Z,X$]
\label{ass:exclusion_restriction_Z_X}
$Z(a) = Z$ and $X(a) = X$ a.s. for all $a \in \mathcal{X}_A$.    
\end{assumption}
\begin{assumption}[Ignorability and positivity given $Z,U,X$]
\label{ass:ignorability_given_Z}
$Y(a) \indep A \, | \, Z,U,X$ for all $a \in \mathcal{X}_A$ and furthermore 
for all $a \in \mathcal{X}_A$: $p(a \, | \, Z,U,X) > 0$  $P_{(Z,U,X)}$-almost surely.
\end{assumption}

\begin{corollary}[Proximal backdoor formula including $Z$] 
    \label{prop:proximal_backdoor_extended}
    Under \cref{ass:valid_outcome_inducing_proxy,ass:valid_treatment_inducing_proxy,ass:outcome_bridge_function,ass:outcome_completeness} and additionally \cref{ass:exclusion_restriction_Z_X,ass:ignorability_given_Z}, we can identify the distribution $\dodistr{Y,Z,X}{}{a}$ as
    \begin{align}
        \label{eq:outcome_bridge_identifying_functional_extended}
        \dodistr{y,z,x}{}{a} = \sum_{w} h(y,w,a,x)\, p(w,z,x).
    \end{align}
\begin{proof}
    We use \cref{ass:valid_outcome_inducing_proxy,ass:valid_treatment_inducing_proxy,,ass:outcome_bridge_function,ass:outcome_completeness} as in \cref{prop:Proximal_backdoor_formula} to arrive at equation \eqref{eq:outcome_bridge_with_u}.
    Using the additional \cref{ass:exclusion_restriction_Z_X,ass:ignorability_given_Z} and consistency, we conclude
    \begin{align*}
        \dodistr{y,z,x}{}{a} &= \sum_{u} p(y \,|\,  a,z,u,x) \, p(z,u,x) \\
                  &= \sum_{u} p(y \,|\,  a,u,x) \, p(z,u,x) && \text{(by \ref{ass:valid_treatment_inducing_proxy})} \\
                  &= \sum_w h(y,w,a,x)\, \sum_u p(w \,|\, u,x) \,p(z, u,x) && \text{(by \eqref{eq:outcome_bridge_with_u}}) \\
                  &= \sum_w h(y,w,a,x)\, \sum_u p(w \,|\, z,u,x) \,p(z, u,x) && \text{(by \ref{ass:valid_outcome_inducing_proxy})} \\
                  &= \sum_{w} h(y,w,a,x)\, p(w,z,x).
    \end{align*}
\end{proof}
\end{corollary}

So far, we have not described a method of identifying interventional distributions containing the outcome-inducing proxy $W$. To this end, we present an alternative approach.

\subsection{Treatment bridge functions}
\label{sec:treatment_bridge_functions}
\citet{Cui2023} introduced treatment bridge functions to identify the moments $\mathbb{E}[Y(a)]$, but note that the treatment bridge function could also be used to identify the interventional distribution $\dodistr{Y}{}{a}$ (see Remark 6 in \citet{Cui2023}). We elaborate on their remark and extend their method to identify the joint distribution $\dodistr{Y,W,X}{}{a}$. Instead of starting from the backdoor formula  \cref{eq:backdoor_formula} as in \cref{sec:Outcome_bridge_functions}, we use the equivalent IPW formula \cref{eq:ipw_formula} with (insufficient) adjustment set $(Z,X)$ and modify it using a treatment bridge function $q$.
\begin{assumption}[Proxy positivity]
\label{ass:proxy_positivity}
For all $a \in \mathcal{X}_A$: $p(a \, |\, W,X) > 0$ $P_{(W,X)}$-almost surely.   
\end{assumption}
\begin{assumption}[Treatment bridge function]
    \label{ass:treatment_bridge_function}
    There exists a treatment bridge function $q(z,a,x)$, solving the integral equation
    \begin{align}
        \label{eq:treatment_bridge_function}
        \sum_{z} q(z,a,x) \, p(z \, | \, w,a,x) = \frac{1}{p(a \, | \, w,x)},
    \end{align}
    in particular requiring positivity \cref{ass:proxy_positivity}. 
\end{assumption}
Sufficient conditions for the existence of a treatment bridge function are given in \citet{Cui2023}. We recall these conditions in \cref{sec:conditions_for_existence_of_bridge_functions} in \cref{auxass:treatment_bridge}. In addition, we require a completeness condition of $W$ with respect to $U$, similar to \cref{ass:outcome_completeness}.
\begin{assumption}[Completeness]
\label{ass:treatment_completeness}
For any $g \in L^2(U)$ and all $a \in \mathcal{X}_A$ and $x \in \mathcal{X}_X$
\begin{align*}
    \mathbb{E}[g(U) \, | \, W,A=a,X=x] = 0 \text{ almost surely} \ \Longleftrightarrow \ g(U) = 0 \text{ almost surely.}
\end{align*}
\end{assumption}

\begin{proposition}[Proximal IPW formula]
    \label{prop:Proximal_IPW_formula}
     Under \cref{ass:valid_treatment_inducing_proxy,ass:valid_outcome_inducing_proxy,ass:latent_ignorability,ass:latent_positivity,ass:proxy_positivity,ass:treatment_completeness,ass:treatment_bridge_function} we can identify the distribution $\dodistr{Y}{}{a}$ as
    \begin{align*}
        \dodistr{y}{}{a} = \sum_{x,z} q(z,a,x) \ p(y,z,a,x).
    \end{align*}
    This result implies the identification formula for the ATE given in Theorem 2.2 of \citet{Cui2023}.
    \begin{proof}
        We express $p(a | w,x)^{-1}$ in two different ways. First, by marginalizing a suitable kernel over $u$, and
        second by using the bridge equation. Assuming the necessary positivity \cref{ass:latent_positivity,ass:proxy_positivity} when dividing by densities, this yields the equation chain
        \begin{align*}
            \sum_u \frac{1}{p(a \, |\, u,x)} \, p(u \, | \, w,a,x)
            = &\sum_u \frac{1}{p(a \, |\, u,w,x)} \, p(u \, | \, w,a,x) && \text{(by \ref{ass:valid_outcome_inducing_proxy})} \\
            = &\sum_u \frac{1}{p(a \, |\, u,w,x)} \, \frac{p(a \, | \, u,w,x) \, p(u \, |\, w,x)}{p(a \, |\, w,x)}  && (\text{Bayes})\\
            = &\frac{1}{p(a \, | \, w,x)} \\
            = &\sum_z q(z,a,x) p(z \, |\, w,a,x) && \text{(by \ref{ass:treatment_bridge_function})} \\
            = &\sum_u \sum_z q(z,a,x) p(z \, |\,a,u,x)  \ p(u \, | \, w,a,x).  && \text{(by \ref{ass:valid_outcome_inducing_proxy})}
        \end{align*}
        Comparing the first and last terms, we may apply the completeness assumption \ref{ass:treatment_completeness} to conclude
        \begin{align}
             \label{eq:treatment_bridge_with_u}
             \frac{1}{p(a \, |\, u,x)} = \sum_z q(z,a,x) \, p(z \, |\, a,u,x).
        \end{align}
        Using the adjustment set $(U,X)$, which is sufficient by \ref{ass:latent_ignorability} and \ref{ass:latent_positivity}, we obtain
        \begin{align*}
            \dodistr{y}{}{a} =  &\sum_{x,u} p(y \, | \, a,u,x) \, p(u,x) =\sum_{x,u} \frac{p(y,a,u,x)}{p(a\,|\,u,x)} \\
            = &\sum_{x,u} p(y,a,u,x) \left(\sum_z q(z,a,x) p(z \, |\,a,u,x)\right) && \text{(by \ref{eq:treatment_bridge_with_u})} \\
            = &\sum_{x,u} \sum_z  \big(p(y \, |\,a,u,x) p(a,u,x) \big) \   q(z,a,x) p(z \, |\,a,u,x) \\
            = &\sum_{x,u}  \sum_z q(z,a,x) \ p(y,z,a,u,x)  && \text{(by \ref{ass:valid_treatment_inducing_proxy})} \\
            = &\sum_{x,z} q(z,a,x) \ p(y,z,a,x).
        \end{align*}
        Taking the expectation of $\dodistr{Y}{}{a}$ for fixed $a$ allows us to identify $\mathbb{E}[Y(a)]$. Taking the difference $\mathbb{E}[Y(1)] - \mathbb{E}[Y(0)]$ recovers the following formula for the ATE that was derived in \citet{Cui2023}
        \begin{align*}
            \text{ATE} = \sum_{a \in \{0,1\}} \sum_{y,z,x} (-1)^{1-a} y  \cdot q(z,a,x) \ p(y,z,a,x).
        \end{align*}
    \end{proof}
\end{proposition}

Analogous to the previous section, additional assumptions allow for the identification of $\dodistr{Y,W,X}{}{a}$ using the treatment bridge function $q$. All three assumptions hold in the example setting from \mbox{\cref{fig:standard_proximal_graph}(d)}. 
\begin{assumption}[Exclusion restrictions for $W,X$]
\label{ass:exclusion_restriction_W_X}
$W(a) = W$ and $X(a) = X$ a.s. for all $a \in \mathcal{X}_A$.    
\end{assumption}
\begin{assumption}[Ignorability and positivity given $W,U,X$]
\label{ass:ignorability_given_W}
$Y(a) \indep A \, | \, W,U,X$ for all $a \in \mathcal{X}_A$ and 
for all $a \in \mathcal{X}_A$: $p(a \, |\, W,U,X) > 0$ $P_{(W,U,X)}$-almost surely.   
\end{assumption}

\begin{assumption}[Conditioning \ref{ass:valid_treatment_inducing_proxy} additionally on $W$]
\label{ass:valid_treatment_inducing_proxy_with_W}
$Z \indep Y \, | \, W,A,U,X$.  
\end{assumption}

\begin{corollary}[Proximal IPW formula including $W$]
    \label{prop:proximal_IPW_formula_extended}
     Under \cref{ass:valid_outcome_inducing_proxy,ass:latent_positivity,ass:proxy_positivity,ass:treatment_bridge_function,ass:treatment_completeness} and additionally the stronger \cref{ass:exclusion_restriction_W_X,ass:ignorability_given_W,ass:valid_treatment_inducing_proxy_with_W}, we identify the distribution $\dodistr{Y,W,X}{}{a}$ as 
    \begin{align}
        \label{eq:treatment_bridge_identifying_functional}
        \dodistr{y,w,x}{}{a} = \sum_{z} q(z,a,x) \ p(y,w,z,a,x).
    \end{align}
    \begin{proof}
        We use \cref{ass:valid_outcome_inducing_proxy,ass:latent_positivity,ass:proxy_positivity,ass:treatment_completeness,ass:treatment_bridge_function} as in \cref{prop:Proximal_IPW_formula} to arrive at \cref{eq:treatment_bridge_with_u}.
        Using the additional \cref{ass:exclusion_restriction_W_X,ass:ignorability_given_W} and consistency, we conclude similarly to \cref{prop:proximal_backdoor_extended}
        \begin{align*}
            \dodistr{y,w,x}{}{a} = &\sum_{u} p(y \, | \, w,a,u,x) \, p(w,u,x)
            = \sum_{u} \frac{p(y,w,a,u,x)}{p(a\,|\,u,x)}  && \text{(by \ref{ass:valid_outcome_inducing_proxy})}\\
            = &\sum_{u} p(y,w,a,u,x) \left(\sum_z q(z,a,x) p(z \, |\,a,u,x)\right)   && \text{(by \eqref{eq:treatment_bridge_with_u})} \\
            = &\sum_{u} \sum_z \big(p(y \, |\,w,a,u,x) p(w,a,u,x) \big) \ q(z,a,x) p(z \, |\,w,a,u,x)    && \text{(by \ref{ass:valid_outcome_inducing_proxy})}\\
            = &\sum_{u,z}  q(z,a,x) \ p(y,w,z,a,u,x) && \text{(by \ref{ass:valid_treatment_inducing_proxy_with_W})}\\
            = &\sum_{z} q(z,a,x) \ p(y,w,z,a,x).
        \end{align*}
    \end{proof}
\end{corollary}

Combining the results from \cref{prop:proximal_backdoor_extended} and \cref{prop:proximal_IPW_formula_extended}, we derived identification results for marginals of the interventional distributions containing either $Z$ or $W$. Still, we have not identified the joint interventional distribution $\dodistr{Y,W,Z,X}{}{a}$ containing both $W$ and $Z$. To this end, we need to introduce new bridge functions and impose further conditions. 

\subsection{Extended bridge functions}
\label{sec:extended_bridge_functions}
We present two novel methods of identifying the joint interventional distribution $\dodistr{Y,W,Z,X}{}{a}$  via \emph{extended outcome bridge functions} or \emph{extended treatment bridge functions}. We define these bridge functions as solutions to the following integral equations.
\begin{assumption}[Extended outcome bridge function]
     \label{ass:extended_outcome_bridge_function}
    There exists an extended outcome bridge function $h(y,w,w',a,x)$, solving the integral equation
    \begin{align}
        \label{eq:extended_outcome_bridge}
        \sum_{w'} h(y,w,w',a,x) \, p(w' \, | \, z,a,x) = p(y,w \, | \, z,a,x).
    \end{align}
\end{assumption}

\begin{assumption}[Extended treatment bridge function]
    \label{ass:extended_treatment_bridge_function}
    There exists an extended treatment bridge function $q(z,z',a,x)$, solving the integral equation
    \begin{align}
        \label{eq:extended_treatment_bridge}
        \sum_{z'} q(z,z',a,x) \, p(z' \, | \, w,a,x) = \frac{p(z \, | \, w,x)}{p(a \, | \, w,x)},
    \end{align}
    in particular requiring positivity \cref{ass:proxy_positivity}. 
\end{assumption}

When comparing integral equation \eqref{eq:extended_outcome_bridge} with the standard outcome bridge equation \eqref{eq:outcome_bridge_function}, we notice two major differences. First, equation \eqref{eq:extended_outcome_bridge} gives an integral representation of $p(y,w \,| \, z,a,x)$, adding dependence on $W$ to the right-hand side compared to the standard outcome bridge equation \eqref{eq:outcome_bridge_function}. Second, the extended outcome bridge function $h$ subsequently also depends on an added parameter $w$ that is not integrated over. Similarly, equation \eqref{eq:extended_treatment_bridge} is an extension of the standard treatment bridge equation \eqref{eq:treatment_bridge_function}, adding explicit dependence on $Z$ to both the right-hand side and the bridge function $q$. While equation \eqref{eq:treatment_bridge_function} gives an integral representation of the inverse propensity score, the extended bridge equation \eqref{eq:extended_treatment_bridge} introduces an additional factor of $p(z \, | \, w,x)$ to the right-hand side. For both \eqref{eq:extended_outcome_bridge} and \eqref{eq:extended_treatment_bridge}, marginalizing over $W$ or $Z$, respectively, recovers the standard bridge equations from \eqref{eq:outcome_bridge_function} and \eqref{eq:treatment_bridge_function}. In particular, if $h$ and $q$ are solutions to the extended bridge equations, then it is easy to see that
\begin{align}
    \label{eq:extension_VS_standard}
    \tilde{h}(y,w,a,x) = \sum_{w'} h(y,w',w,a,x) \quad \text{and} \quad 
    \tilde{q}(z,a,x) = \sum_{z'} q(z',z,a,x)
\end{align}
will be valid solutions to the standard bridge equations. Accordingly, it is appropriate to call these bridge functions \emph{extended} relative to standard outcome and treatment bridge functions.
 
We construct solutions to these extended bridge equations in the case of discrete random variables. In the discrete case, \eqref{eq:extended_outcome_bridge} and \eqref{eq:extended_treatment_bridge} are matrix equations that admit unique solutions if we assume the invertibility of $p(W \, | \, Z,a,x)$ and $p(Z \, | \, W,a,x)$. Using similar notation as \citet{Miao2018}, the solutions can be written as
\begin{align}
    \label{eq:discrete_extended_outcome_bridge}
    h(y,w,w',a,x) &= \sum_{z'}  p(y,w \, |\, z',a,x) (p^{-1}(W \, | \, Z,a,x))_{z',w'}, \\
    \label{eq:discrete_extended_treatment_bridge}
    q(z,z',a,x) &= \sum_{w'} \frac{p(z \, | \, w',x)}{p(a \, | \, w',x)} (p^{-1}(Z \, | \, W,a,x))_{w',z'},
\end{align} 
where $(p^{-1}(W \, | \, Z,a,x))_{z',w'}$ is the $(z',w')$ entry of the inverse matrix of $p(W \, | \, Z,a,x)$, and similarly for $(p^{-1}(Z \, | \, W,a,x))_{w',z'}$. With a slight abuse of notation, we also represent these bridge functions as matrix products $ h(y,w,W,a,x) = p(y,w \, |\, Z,a,x) p^{-1}(W \, | \,Z,a,x)$ and $q(z,Z,a,x) = B(a,z,x,W) p^{-1}(Z \, | \, W,a,x)$ with vector $B$ defined via $B(a,z,x,w) = p(z \, | \, w,x) / p(a \,|\, w,x)$.

Using either the extended outcome bridge function $h$ or the extended treatment bridge function $q$, we can prove identification of $\dodistr{Y,W,Z,X}{}{a}$. Following the presentation style of \citet{Miao2018} and \citet{Cui2023}, we first present our new proximal identification results for discrete random variables using matrix notation. We impose the following assumptions, which are satisfied in our example setting. Notably, these assumptions are strictly stronger than \cref{ass:exclusion_restriction_Z_X,ass:ignorability_given_Z} or \cref{ass:exclusion_restriction_W_X,ass:ignorability_given_W}.
\begin{assumption}[Exclusion restrictions for $W,Z,X$]\hfill 
\label{ass:exclusion_restriction_W_Z_X}
$W(a) = W$, $Z(a) = Z$ and $X(a) = X$ a.s. for all $a \in \mathcal{X}_A$.    
\end{assumption}

\begin{assumption}[Ignorability and positivity given $W,Z,U,X$]
\label{ass:ignorability_given_W_Z}
$Y(a) \indep A \, | \, W,Z,U,X$ for all $a \in \mathcal{X}_A$ and furthermore
for all $a \in \mathcal{X}_A$: $p(a \, |\, W,Z,U,X) > 0$ $P_{(W,Z,U,X)}$-almost surely.   
\end{assumption}

\begin{corollary}[Interventional joint identification I]
    \label{cor:identify_entire_joint_with_outcome_bridge}
    Let $A,Y,W,Z,X$ and $U$ be discrete random variables. If we assume that the matrix $p(W \, |\, Z,a,x)$ is invertible with inverse $p^{-1}(W \, | \, Z,a,x)$, the joint distribution $\dodistr{Y,W,Z,X}{}{a}$ is identifiable under \cref{ass:ignorability_given_W_Z,ass:valid_outcome_inducing_proxy,ass:valid_treatment_inducing_proxy_with_W,ass:exclusion_restriction_W_Z_X,ass:outcome_completeness} with
    \begin{align*}
        \dodistr{y,w,z,x}{}{a} = h(y,w,W,a,x) \, p(W,z,x) \qquad \text{and $h$ as in \eqref{eq:discrete_extended_outcome_bridge}}.
    \end{align*}
    \begin{proof}
        By the contraction property of conditional independence, \cref{ass:valid_outcome_inducing_proxy,ass:valid_treatment_inducing_proxy_with_W,ass:ignorability_given_W_Z} imply 
        \begin{alignat}{2}
            Z &\indep (Y,W)  | \, A,U,X   &&\text{by \ref{ass:valid_outcome_inducing_proxy} and \ref{ass:valid_treatment_inducing_proxy_with_W},} 
            \label{eq:valid_treatment_inducing_proxy_as_Y_prime}\\
            A  &\indep (Y(a),W) | \, Z,U,X \qquad &&\text{by \ref{ass:valid_outcome_inducing_proxy} and \ref{ass:ignorability_given_W_Z}.} 
            \label{eq:ignorability_for_Y_prime}
        \end{alignat}
        Using the matrix notation outlined in \cref{sec:notation}, we derive two equivalent expressions for $p(y,w \,| \, Z,a,x)$ similar to the proof of \cref{prop:Proximal_backdoor_formula}.
        \begin{align*}
            p(y,w \, |\, U,a,x) p(U \, | \, Z,a,x) &=p(y,w \, |\, Z,a,x) && \text{(by \ref{eq:valid_treatment_inducing_proxy_as_Y_prime})}\\
            &= p(y,w \, |\, Z,a,x)  \, p^{-1}(W \, |\, Z,a,x) \, p(W \, |\, Z,a,x) \\
            &= p(y,w \, |\, Z,a,x)  p^{-1}(W \, |\, Z,a,x) \, p(W \, |\, U,x) p(U \, | \,Z,a,x). && \text{(by \ref{ass:valid_outcome_inducing_proxy})}
        \end{align*}
        With \cref{ass:outcome_completeness} and \cref{cor:discrete_completeness_matrix_rank}, we conclude that $p(U \, |\, Z,a,x)$ has full row-rank, and therefore 
        \begin{align}
            \label{eq:matrix_bridge_equation_with_W}
            p(y,w \, |\, U,a,x) = h(y,w,W,a,x) \, p(W \, |\, U,x),
        \end{align}
        for $h(y,w,W,a,x) = p(y,w \, |\, Z,a,x)  p^{-1}(W \, |\, Z,a,x)$.
        Finally, we can use \eqref{eq:ignorability_for_Y_prime}, \ref{ass:exclusion_restriction_W_Z_X}, the positivity \cref{ass:ignorability_given_W_Z}, and consistency to conclude
        \begin{align*}
            \dodistr{y,w,z,x}{}{a} &= p(y,w  \, |\, U,a,x) \, p(U,z,x) && \text{(by \ref{eq:ignorability_for_Y_prime} and \ref{eq:valid_treatment_inducing_proxy_as_Y_prime})}\\
                          &= h(y,w,W,a,x) \, p(W \, |\, U,x) \, p(U,z,x) && \text{(by \ref{eq:matrix_bridge_equation_with_W})} \\
                          &= h(y,w,W,a,x) \, p(W, z,x). && \text{(by \ref{ass:valid_outcome_inducing_proxy})}
        \end{align*}        
        \end{proof}
\end{corollary}
In the general case of non-discrete random variables, we need to impose the existence of a solution to the integral equation \eqref{eq:extended_outcome_bridge} as an additional assumption. The previous proof can then be slightly altered to allow for non-discrete random variables (see \cref{sec:extended_bridge_functions_for_general_state_spaces}). However, we do not provide general conditions for the existence of such a solution here.\footnote{We suspect that sufficient conditions can be stated similarly to the regularity conditions \ref{auxass:outcome_bridge} and \ref{auxass:treatment_bridge} used for standard outcome and treatment bridge functions.}

Comparing the proofs of \cref{cor:identify_entire_joint_with_outcome_bridge} and \cref{prop:Proximal_backdoor_formula}, we can see that \cref{cor:identify_entire_joint_with_outcome_bridge} indirectly defines a new outcome $Y'=(Y,W)$ and then applies a discrete version of the identification result \cref{prop:proximal_backdoor_extended}. In this sense, equations \eqref{eq:ignorability_for_Y_prime} and \eqref{eq:valid_treatment_inducing_proxy_as_Y_prime} are restating \cref{ass:ignorability_given_Z,ass:valid_treatment_inducing_proxy} for the new outcome variable $Y'$. Using the same assumptions, we can also extend \cref{prop:proximal_IPW_formula_extended} to recover the full interventional joint in the discrete case using an extended treatment bridge function. As before, we first discuss the discrete case.

\begin{corollary}[Interventional joint identification II]
    \label{cor:identify_entire_joint_with_treatment_bridge}
    Let $A,Y,W,Z,X$ and $U$ be discrete random variables. If we assume that the matrix $p(Z \, |\, W,a,x)$ is invertible with inverse $p^{-1}(Z \, |\, W,a,x)$, the joint distribution $\dodistr{Y,W,Z,X}{}{a}$ is identifiable under \cref{ass:ignorability_given_W_Z,ass:valid_outcome_inducing_proxy,ass:valid_treatment_inducing_proxy_with_W,ass:exclusion_restriction_W_Z_X,ass:proxy_positivity,ass:treatment_completeness} via 
    \begin{align*}
        \dodistr{y,w,z,x}{}{a} = q(z,Z,a,x) \, p(Z,y,a,w,x) \qquad \text{and $q$ as in \eqref{eq:discrete_extended_treatment_bridge}}.
    \end{align*}
    \begin{proof}
        By Bayes theorem and positivity \cref{ass:ignorability_given_W_Z,ass:proxy_positivity}, we conclude
        \begin{align*}
            \sum_{u} \frac{p(z \, | \, u,x)}{p(a \, |\, u,x)} p(u \, | \, w,a,x) &= \sum_{u} \frac{p(z \, | \, u,x)}{p(a \, |\, u,x)} \frac{p(a \, | \, w,u,x) \, p(u \, | \, w,x)}{p(a \, | \, w,x)} && \text{(Bayes)} \\
            &= \sum_{u} \frac{p(z \, | \, w, u,x)}{p(a \, |\, w,u,x)} \frac{p(a \, | \, w,u,x) \, p(u \, | \, w,x)}{p(a \, | \, w,x)}&& \text{(by \ref{ass:valid_outcome_inducing_proxy})} \\
            &= \frac{p(z \, | \, w,x)}{p(a \, | \, w,x)}.
        \end{align*}
        Now we return to the matrix notation and define the vectors
        \begin{align*}
            A(a,z,x,U) = \left(\frac{p(z \, | \, u_i,x)}{p(a \, | \, u_i,x)} \right)_{i=1}^{n_U} \quad ; \quad B(a,z,x,W) =  \left(\frac{p(z \, | \, w_i,x)}{p(a \, | \, w_i,x)} \right)_{i=1}^{n_W},
        \end{align*}
        where $\mathcal{X}_U = \{u_1, \dots, u_{n_U}\}$ and $\mathcal{X}_W = \{w_1, \dots, w_{n_W}\}$. In this notation,  $A \, p(U \,| \, W,a,x) = B$. With \cref{ass:valid_outcome_inducing_proxy}, we have
        \begin{align*}
           A(a,z,x,U) \, p(U \,| \, W,a,x) = B(a,z,x,W) &= B(a,z,x,W)  \,p^{-1}(Z \, | \, W,a,x) \, p(Z \, | \, W,a,x) \\
           &= B(a,z,x,W)  \,p^{-1}(Z \, | \, W,a,x) \, p(Z \, | \, U,a,x) \, p(U \,| \, W,a,x).
        \end{align*}
        With \cref{ass:treatment_completeness} and \cref{cor:discrete_completeness_matrix_rank}, we conclude that $p(U \, |\, W,a,x)$ has full row-rank, and therefore
        \begin{align}
            \label{eq:matrix_bridge_equation_with_Z}
            A(a,z,x,U)= B(a,z,x,W) \,p^{-1}(Z \, | \, W,a,x) \, p(Z \, | \, U,a,x) = q(z,Z,a,x) \, p(Z \, | \, U,a,x),
        \end{align}
        for $q(z,Z,a,x) = B(a,z,x,W) p^{-1}(Z \, | \, W,a,x)$. Finally, we can use \cref{ass:ignorability_given_W_Z,ass:exclusion_restriction_W_Z_X} and consistency to conclude
        \begin{align*}
            \dodistr{y,w,z,x}{}{a} &= \sum_{u} p(y \, | \, a,w,z,u,x) \, p(w,z,u,x) \\
            &= \sum_{u} p(y \, | \, a,w,u,x) \,p(w,u,x)\, p(z \,| u,x) && \text{(by \ref{ass:valid_outcome_inducing_proxy} and \ref{ass:valid_treatment_inducing_proxy_with_W})} \\
            &= \sum_{u} p(y,a,w,u,x) \,\frac{p(z \, | \, u,x)}{p(a \, |\, u,x)} = A(a,z,x,U) p(U,y,a,w,x)  && \text{(by \ref{ass:valid_outcome_inducing_proxy})}\\
            &= q(z,Z,a,x) \, p(Z \, | \, U,a,x) p(U,y,a,w,x) && \text{(by \ref{eq:matrix_bridge_equation_with_Z})} \\
            &=  q(z,Z,a,x) \, p(Z \, | \, U,a,w,x) p(U,y,a,w,x) && \text{(by \ref{ass:valid_outcome_inducing_proxy})}  \\
            &= q(z,Z,a,x) \, p(Z,y,a,w,x). && \text{(by \ref{ass:valid_treatment_inducing_proxy_with_W})} 
        \end{align*}
        \end{proof}
\end{corollary}

In the general setting, one must again additionally add the existence of a solution to the integral equation \eqref{eq:extended_treatment_bridge} as an assumption, leading to a similar proof with altered notation (see \cref{sec:extended_bridge_functions_for_general_state_spaces}). As before, we do not provide conditions for the existence of these solutions at this point.

We can recover analogs of the identifying functionals in \cref{prop:proximal_backdoor_extended,prop:proximal_IPW_formula_extended} from the extended identifying functionals in \cref{cor:identify_entire_joint_with_outcome_bridge,cor:identify_entire_joint_with_treatment_bridge} by marginalization and using \eqref{eq:extension_VS_standard}. In this sense, we can use extended bridge functions to generalize the previous results. Notably, we do require stronger assumptions for this, most notably the existence of a solution to the extended bridge equations. Accordingly, the application of extended bridge functions does not replace the previous methods but should instead be regarded as an additional tool for proximal identification. There are many possible applications for extended bridge functions that can be explored. We discuss some of these applications in \cref{sec:summary}.

\section{Proximal identification theory}
\label{sec:proximal_identification_theory}
In the following, we apply the previous results to a broader class of causal graphs. In particular, we explore which identification results can be obtained by combining the identification results using outcome bridge functions, treatment bridge functions, and extended bridge functions from \cref{sec:outcome_and_treatment_bridge_functions} with standard identification theory. \citet{Shpitser2021} presented a Proximal ID Algorithm that specifically uses \textit{outcome} bridge functions together with fixing operations. We explore how adding \textit{treatment} bridge functions as well as extended bridge functions can expand on these results. Finally, we introduce a generalized framework for proximal ID algorithms in \cref{thm:generalized_proximal_ID_algorithm}. We begin by presenting a motivating example that illustrates the use of treatment bridge functions in conjunction with the standard g-formula in a front-door setting. This yields a new proximal front-door identification strategy that complements the proximal front-door criterion introduced in \citet{Shpitser2021}. 

\subsection{Proximal front-door criterion}
\label{Subsec:Proximal_front_door}
We return to the example medical study introduced in \cref{sec:Introduction} and additionally measure an inflammatory acute-phase parameter $M$ of a patient, for example, the C-reactive protein (CRP) \citep{Sproston2018}. We assume that $M$ is a perfect measurement of the inflammatory process that is part of an infection. Our causal assumption is that the effect of the influenza vaccination $A$ on hospitalization $Y$ is through suppressing an infection and the inflammatory process. Thus, $M$ is at least partially mediating the causal effect of $A$ on $Y$ through an $A \rightarrow M \rightarrow Y$ pathway.  

In our example setting, traumatic injuries also routinely lead to increased CRP values \citep{Sproston2018}, forcing us to include a $W \rightarrow M$ pathway. Furthermore, chronic inflammatory processes and fever presenting with raised CRP values may cause a general practitioner to perform a routine cancer screening on a patient.\footnote{ We acknowledge that the connection between inflammation, cancer, and cancer screenings is much more complex. For the sake of argument, we only consider $M \rightarrow Z$.} Thus, we also want to include a $M \rightarrow Z$ pathway. Critically, we assume that $U$ does not have a direct effect on the mediator $M$. This naturally leads us to consider the proximal front-door causal graph (a) in \cref{fig:proximal_front_door_graph}. Compared to the setting without a mediator in \cref{fig:standard_proximal_graph}, we added the mediator and its associated pathways and changed the direction of the causal relation between $Z$ and $A$ to keep the graph acyclic. This makes $Z$ a post-treatment proxy of $U$. We refer to \citet{TchetgenTchetgen2024} for a discussion of post- and pre-treatment proxies in proximal causal inference and to \citet{Ringlein2025} for practical examples. 

\begin{figure}[h!]
        \centering
        \includegraphics[width=0.9\linewidth]{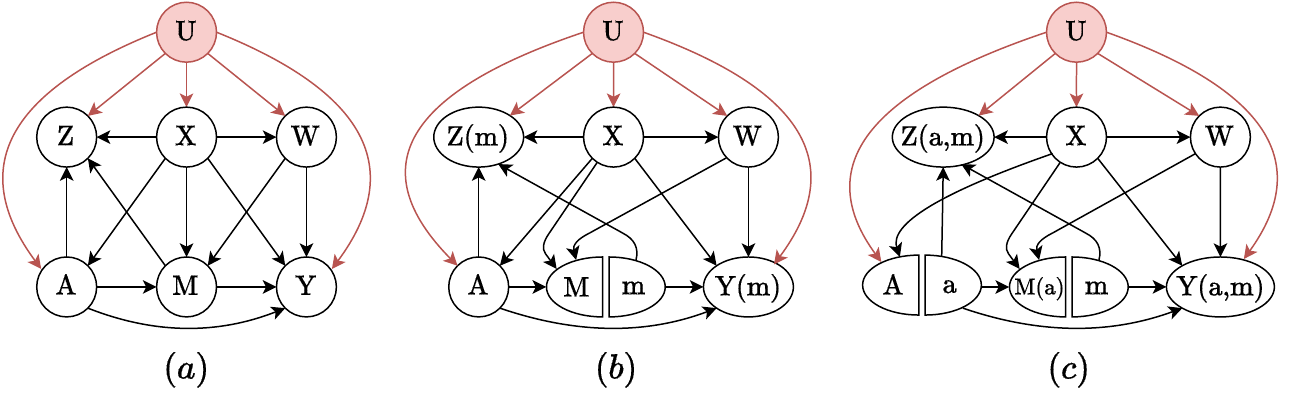}
        \caption{\textbf{(a)} A causal diagram representing the model used for proximal front-door identification with treatment bridge functions. We prove an identification formula for $\dodistr{y}{}{a}$ in \cref{thm:proximal_front_door_treatment_bridge}.  \textbf{(b)} A SWIG representing the causal model in $(a)$ after the intervention $\text{do}(M = m)$. \textbf{(c)} A SWIG representing the causal model in $(a)$ after the joint intervention $\text{do}(A=a, M = m)$.}
        \label{fig:proximal_front_door_graph}
\end{figure}
The identification strategies presented in \cref{sec:outcome_and_treatment_bridge_functions} are not directly applicable to this extended framework. In particular, the $W \rightarrow M \rightarrow Z$ path renders the required conditional independence $W \indep (Z,A) \, | \, U,X$ from \cref{ass:valid_outcome_inducing_proxy} invalid. To address this problem, we first perform an intervention $\text{do}(M = m)$ and deduce the corresponding interventional distribution $\dodistr{y,a,w,z,x}{}{m}$ using standard identification arguments. This fixes $M$ at a constant value and blocks the $W \rightarrow M \rightarrow Z$ path. We then require the corresponding assumptions from \cref{sec:outcome_and_treatment_bridge_functions} to hold for the kernel $\dodistr{Y,A,W,Z,X}{}{m}$ after the intervention $\text{do}(M=m)$. A similar calculation as in \cref{prop:Proximal_IPW_formula} allows us to deduce the joint interventional distribution $\dodistr{Y,W,X}{}{a,m}$ after the combined intervention $\text{do}(M=m, A=a)$. Lastly, suitable exclusion restrictions enable identification of $\dodistr{Y}{}{a}$ from the joint interventional kernel $\dodistr{Y,W,X}{}{a,m}$. We require the following assumptions that can be deduced from the SWIGs presented in \mbox{\cref{fig:proximal_front_door_graph}(b) and (c)}. Analogous to \cite{Shpitser2021}, we will not explicitly elaborate on positivity assumptions going forward.
\begin{assumption}
\label{ass:kernel_district_decomposition}
$M(a) \indep (Y(a,m), A) \, | \, W,X$ for all $m \in \mathcal{X}_M$ and $a \in \mathcal{X}_A$.
\end{assumption}

\begin{assumption}[Exclusion restrictions]
\label{ass:kernel_exclusion_restrictions_M}
$A(m) = A$, $W(m,a) = W$, $X(m,a) = X$, $U(m) = U$ almost surely for all $m \in \mathcal{X}_M$ and $a \in \mathcal{X}_A$.
\end{assumption}

\begin{assumption}[Ignorability for $M$]
\label{ass:kernel_ignorability_M}
$(Y(m), Z(m)) \indep M \,|\, A,W,X$ for all $m \in \mathcal{X}_M$.
\end{assumption}

The following assumptions are restatements of the assumptions of \cref{prop:proximal_IPW_formula_extended} that allow us to apply identification results using treatment bridge functions to the kernel $\dodistr{y,a,w,z,x}{}{m}$. 
\begin{assumption}[Kernel latent ignorability]
\label{ass:kernel_latent_ignorability}
$Y(a,m) \indep A \, | \, W,U,X$ for all $a \in \mathcal{X}_A, m \in \mathcal{X}_M$.    
\end{assumption}
\begin{assumption}[Kernel valid outcome-inducing proxy]
\label{ass:kernel_valid_outcome_inducing_proxy}
$W \indep (Z(m),A) \, | \, U,X$.
\end{assumption}
\begin{assumption}[Kernel valid treatment-inducing proxy]
\label{ass:kernel_valid_treatment_inducing_proxy}
$Y(m) \indep Z(m) \, | \, W,A,U,X.$
\end{assumption}

\begin{assumption}[Kernel completeness]
\label{ass:kernel_treatment_completeness}
For any $g \in L^2(U)$ and all $a \in \mathcal{X}_A,m \in \mathcal{X}_M$ and $x\in \mathcal{X}_X $
\begin{align*}
    \mathbb{E}_{\dodistr{\cdot}{}{m}}[g(U) \, | \, W,A=a,X=x] = 0 \text{ almost surely}  \ \Longleftrightarrow \ g(U) = 0 \text{ almost surely}.
\end{align*}
\end{assumption}

\begin{assumption}[Kernel treatment bridge function]
    \label{ass:kernel_treatment_bridge_function}
    There exists a treatment bridge function $q^{(m)}(z,a,x)$, solving the integral equation
    \begin{align}
        \label{eq:kernel_treatment_bridge_function}
        \sum_{z} q^{(m)}(z,a,x) \, \dodistr{z}{a,w,x}{m} = \frac{1}{\dodistr{a}{w,x}{m}},
    \end{align}
    in particular requiring positivity of $\dodistr{a}{w,x}{m}$.
\end{assumption}

\begin{theorem}[Proximal front-door formula with treatment bridge function] 
\label{thm:proximal_front_door_treatment_bridge}
Under the \cref{ass:kernel_district_decomposition,ass:kernel_exclusion_restrictions_M,ass:kernel_ignorability_M,ass:kernel_latent_ignorability,ass:kernel_treatment_bridge_function,ass:kernel_treatment_completeness,ass:kernel_valid_outcome_inducing_proxy,ass:kernel_valid_treatment_inducing_proxy}, we can identify the distribution of $Y(a)$ as
    \begin{align}
        \dodistr{y}{}{a} =   \sum_{m,x,z}  q^{(m)}(z,a,x) \, p(y,m,z,a,x).
    \end{align}

    \begin{proof}
        We start by decomposing the causal query similar to the standard front-door criterion. By consistency, positivity, and \cref{ass:kernel_exclusion_restrictions_M,ass:kernel_district_decomposition}, we obtain
        \begin{align}
            \dodistr{y}{}{a} = &\sum_{m,w,x} \dodistr{y,m}{w,x}{a} p(w,x) && \text{(by \ref{ass:kernel_exclusion_restrictions_M})}\nonumber\\
            = &\sum_{m,w,x} \dodistr{y}{w,x}{a,m} \dodistr{m}{w,x}{a} p(w,x). && \text{(by \ref{ass:kernel_district_decomposition} and \ref{ass:kernel_exclusion_restrictions_M})} 
            \label{eq:prox_front-door_treatment_bridge_result_1}
            \end{align}
        It remains to identify the two kernels $\dodistr{m}{w,x}{a}$ and $\dodistr{y,w,x}{}{a,m}$. With \cref{ass:kernel_district_decomposition}, we can directly conclude $\dodistr{m}{w,x}{a} = p( m \, | \, a,w,x)$. In order to identify $\dodistr{y,w,x}{}{a,m}$, we proceed in two steps. First, we identify the kernel $\dodistr{y,a,w,z,x}{}{m}$ starting from the observational distribution. By \cref{ass:kernel_exclusion_restrictions_M,ass:kernel_ignorability_M} and consistency, we have
        \begin{align}
              \label{eq:m_intervention_distribution_factorization}
              \dodistr{y,a,w,z,x}{}{m} = p(y,z \, | \, m,a,w,x) p(a,w,x).
        \end{align}
        Now we can use $\dodistr{y,a,w,z,x}{}{m}$ to identify $\dodistr{y,w,x}{}{a,m}$ with a treatment bridge method similar to \cref{prop:proximal_IPW_formula_extended}. We can go through the same derivation as in \cref{prop:Proximal_IPW_formula} and use \cref{ass:kernel_valid_outcome_inducing_proxy,ass:kernel_treatment_completeness,ass:kernel_treatment_bridge_function} to show
        \begin{align*}
            \frac{1}{\dodistr{a}{u,x}{m}} = \sum_z q^{(m)}(z,a,x)  \, \dodistr{z}{u,a,x}{m}. 
        \end{align*}

        Combining this result with consistency, ignorability \cref{ass:kernel_latent_ignorability}, and \cref{ass:kernel_valid_outcome_inducing_proxy,ass:kernel_valid_treatment_inducing_proxy}, we can follow the proof of \cref{prop:proximal_IPW_formula_extended} to arrive at
        \begin{align*}
            \dodistr{y,w,x}{}{a,m} = \sum_z q^{(m)}(z,a,x) \, \dodistr{y,w,z,a,x}{}{m}.
        \end{align*}
        Thus, we have successfully identified the kernel $\dodistr{y,w,x}{}{a,m}$. We trivially condition the identified kernels on $W$ and $X$ to obtain an identifying functional for $\dodistr{y}{w,x}{a,m}$. Finally, we substitute our results for $\dodistr{y,a,z}{w,x}{m}$ and $\dodistr{y,w,x}{}{a,m}$ into     \eqref{eq:prox_front-door_treatment_bridge_result_1} and conclude
        \begin{align*}
            \dodistr{y}{}{a} &= 
            \sum_{m,w,x}  \left( \sum_z q^{(m)}(z,a,x) \dodistr{y,z,a}{w,x}{m} \right)p( m \, | \, a,w,x) \, p(w,x) \\
            &= \sum_{m,w,z,x} \left(q^{(m)}(z,a,x) \ p(y,z \, | \, m,a,w,x)p(a \, | \,w,x) \right)  \ p( m \, | \, a,w,x)\, p(w,x)  \quad \text{(by \eqref{eq:m_intervention_distribution_factorization})} \\
            &= \sum_{m,x,z}  q^{(m)}(z,a,x) \, p(y,m,z,a, x).
            \end{align*}
    \end{proof}
\end{theorem}

From the proof, we observe that initially conditioning on the outcome proxy $W$ was a crucial step in splitting the identification problem into smaller subproblems in equation \eqref{eq:prox_front-door_treatment_bridge_result_1}. This allowed us to first identify the kernels $\dodistr{y}{w,x}{a,m}$ and $\dodistr{m}{w,x}{a}$ separately and later combine them to finally identify $\dodistr{y}{}{a}$. When comparing the conditional independence from \cref{ass:kernel_district_decomposition} to the SWIG in \mbox{\cref{fig:proximal_front_door_graph}(c)}, we can see that a similar splitting would not be possible without conditioning on $W$. This observation illustrates that, even when the primary objective is to identify the interventional kernel $\dodistr{y}{}{a}$, which does not explicitly involve the proxy variables $W$ or $Z$, it can still be advantageous to identify interventional distributions that do include $W$ or $Z$. This hints at the relevance of the results in  \cref{prop:proximal_backdoor_extended,prop:proximal_IPW_formula_extended,cor:identify_entire_joint_with_outcome_bridge,,cor:identify_entire_joint_with_treatment_bridge} for a general proximal identification theory since these methods allow us to identify joint interventional distributions containing some or all of the proxy variables.

Several different proximal front-door settings have been considered in the literature. \citet{Shpitser2021} also consider a proximal front-door graph containing a $Z \rightarrow M \rightarrow W$ pathway instead (see \mbox{\cref{fig:proximal_front_door_graph_comparison}(b)}) and provide an identification formula using an outcome bridge function. This identification result does not apply in our setting since the core assumption $Y(a,m) \indep M(a) \, | \, Z,X$ is violated. As we discuss in \cref{sec:applying_the_generalized_proximal_ID_algorithm}, the graph considered in this section can also not be identified using the original proximal ID algorithm from \citet{Shpitser2021}. \citet{Dukes2023} presented an identification formula in a different proximal front-door setting where a confounding effect on the mediator $M$ via $U\rightarrow M$ is considered, assuming no direct effects between the proxies $Z$ and $W$ and the mediator $M$ (see \mbox{\cref{fig:proximal_front_door_graph_comparison}(c)}). The setting is also expanded in a different direction by \citet{Ghassami2024} to unmeasured mediators, assuming the existence of $U$-unconfounded proxies $W$ and $Y$ (see \mbox{\cref{fig:proximal_front_door_graph_comparison}(d)}). We refer to Section 7 of \citet{Ghassami2024} for a comparison of the different front-door models and the corresponding identification strategies. 
 \begin{figure}[h!]
        \centering
        \includegraphics[width=1.0\linewidth]{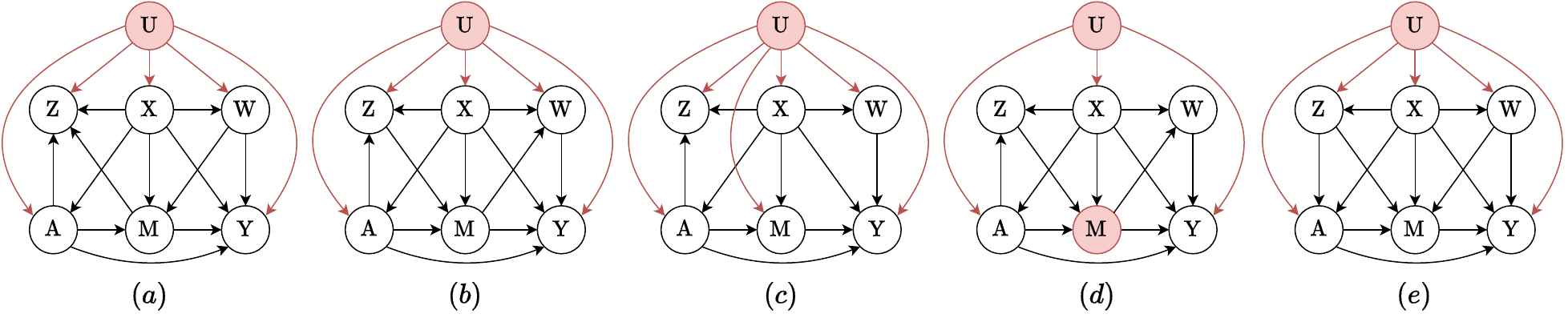}
        \caption{Comparison of different proximal front-door graphs: \textbf{(a)} A proximal front-door graph that is consistent with the assumptions of \cref{thm:proximal_front_door_treatment_bridge}. In this graph, $\dodistr{y}{}{a}$ cannot be identified using the methods from \citet{Shpitser2021}. We present an identification result of $\dodistr{y}{}{a}$ using a treatment bridge function in \cref{thm:proximal_front_door_treatment_bridge}. \textbf{(b)} A proximal front-door graph for which $\dodistr{y}{}{a}$ was shown to be identifiable using an outcome bridge function in \citet{Shpitser2021}. \textbf{(c)} A causal diagram consistent with the assumptions of the model in \citet{Dukes2023}.
        \textbf{(d)} A causal diagram consistent with the assumptions of the model in \citet{Ghassami2024}. \textbf{(e)} A proximal front-door graph that we consider in \cref{sec:applying_the_generalized_proximal_ID_algorithm} and for which we show identification of $\dodistr{y}{}{a}$ via extended bridge functions.} 
        \label{fig:proximal_front_door_graph_comparison}
\end{figure}

\subsection{Proximal identification algorithms}
\label{sec:Identification_algorithms}
In the previous section, we demonstrated that combining conventional identification strategies, such as the front-door criterion, with treatment bridge functions can lead to additional identification results. In this section, we present a generalized framework for proximal identification strategies using a new formulation based on kernel operations. This generalizes the proximal ID algorithm from \citet{Shpitser2021} and allows for identification strategies using a combination of proximal identification methods. 

While much of our previous discussion was guided by the example medical study, we now return to the general setting of nonparametric identification. We formulate the proximal identification problem similar to \cref{prob:nonparametric_identification_problem}. Importantly, we also have to include non-graphical assumptions, such as completeness or solvability of the bridge equations. Given a fixed observational distribution $p(V)$, the results in \cref{sec:outcome_and_treatment_bridge_functions} show that suitably chosen non-graphical assumptions can substantially restrict the class of causal models compatible with $p(V)$. In particular, these assumptions reduce the original set of causal models that yield $p(V)$ as their observational distribution down to a subset of models that not only agree on $p(V)$ but also induce the same interventional distribution. Under this restriction, the interventional distribution is then uniquely identified from $p(V)$.
\begin{boxA}
\begin{problem}[Nonparametric proximal identification problem] 
\label{prob:nonparametric_proximal_identification_problem}
\par
{\leftskip=0.em
Consider a causal model with observed variables $V$ and latent variables $U$.
Assume that the model induces the causal graph $\mathcal{G}(V \cup U)$ and obeys a given set of non-graphical assumptions required for proximal causal inference.
Fix an intervention set $A \subseteq V$ and an outcome set $Y \subseteq V \setminus A$.
Under these assumptions, are the interventional distributions $\dodistr{Y}{}{a}$
uniquely determined by the observational distribution $p(V)$?
\par}
\end{problem}
\end{boxA}

In the following, we consider a general causal graph $\mathcal{G}(V \cup U)$ with observed variables $V$ and unobserved variables $U$. We fix two disjoint sets $A,Y \subseteq V$ and try to identify $\dodistr{Y}{}{a}$ for $a \in \mathcal{X}_A$ using an identification algorithm that includes proximal causal inference methods. We break down the identification algorithm into two building blocks. 

\subsubsection{District factorization}
\label{sec:district_factorization}
First, we define the input and target of the identification algorithm. Every identification strategy must start with the joint observational distribution $p(V)$ as input. The fundamental objects of any identification strategy are interventional kernels of the general form $\dodistr{A}{B}{C}$.\footnote{We note that any interventional kernel becomes meaningful only in the context of the causal model from which it was derived. In this sense, the fundamental objects are kernels combined with the underlying causal model.} Starting from the observational kernel $p(V)$, we iteratively identify new kernels to ultimately reach $\dodistr{Y}{}{A}$. 

The identification problem is greatly simplified using \textit{district factorization}. Instead of directly targeting $\dodistr{Y}{}{a}$, we try to identify multiple intermediate kernels. First, we apply a latent projection to $\mathcal{G}(V \cup U)$ and project out a set $H \supseteq U$, removing the unmeasured variables $U$ from the graph and potentially some measured variables $V \cap H$ as well. This yields the ADMG $\mathcal{G}(V \setminus H)$ with $V \setminus H$ as vertex set. We recall ADMGs and latent projections in \cref{sec:ADMGs_and_latent_projections}. 

We can further reduce the set of random variables in our graph by only considering those random variables that causally affect $Y$ after intervening on $A$. Let $Y^*$ be the set of ancestors of $Y$ in $\mathcal{G}(V \setminus H)$ via directed paths that do not intersect $A$. We define the subgraph $\mathcal{G}(V \setminus H)_{Y^*}$ on the vertex set $Y^*$. The districts of the ADMG $\mathcal{G}(V \setminus H)_{Y^*}$ allow for the following factorization.
\begin{lemma}[District factorization \citep{Tian2002}] 
    \label{lem:district_factorization}
    For a district $D \in \mathcal{D}(\mathcal{G}(V \setminus H)_{Y^*})$ let $s_D$ be a value assignment to $\pa(D) \setminus D$ consistent with the intervention $A=a$. Then the interventional distribution factorizes as
    \begin{align}
        \dodistr{Y}{}{a} = \sum_{Y^* \setminus Y} \prod_{D \in \mathcal{D}(\mathcal{G}(V \setminus H)_{Y^*})} \dodistr{D}{}{s_D}.
    \end{align}
    \begin{proof}
        We refer to \citet{Tian2002} for a proof and discussion of the result.
    \end{proof}
\end{lemma}

District factorization shows that the identification of each $\dodistr{D}{}{s_D}$ is a sufficient condition for the identification of $\dodistr{Y}{}{a}$. The ID algorithms described in \citet{Shpitser2006} and \citet{Shpitser2021} implement this task by targeting the stronger kernel $\dodistr{D}{}{V^* \setminus D}$,
where $V^* = V \setminus H$. In the NPSEM framework, this is licensed by the exclusion restriction $D(s_D) = D(t_{V^* \setminus D})$ for any value assignment $t_{V^* \setminus D}$ on $V^* \setminus D$ consistent with the intervention $A = a$ and the intervention $s_D$ on $\pa(D) \setminus D$. Identifying the larger kernels $\dodistr{D}{}{V^* \setminus D}$ is therefore also a sufficient condition for identification of $\dodistr{Y}{}{a}$. This condition was also shown to be necessary for identification in \cref{prob:nonparametric_identification_problem} \citep{Shpitser2006,Huang2006}. All identification algorithms that we discuss below target $\dodistr{D}{}{V^* \setminus D}$ instead of $\dodistr{Y}{}{a}$ and we will refer to the $\dodistr{D}{}{V^* \setminus D}$ kernels as \textit{target kernels}. 

Notably, we deliberately treat the variables in the intersection $H\cap V$  as if they were unobserved, and consequently eliminate them from the graph via latent projection. Since the proximal causal inference methods from \cref{prop:proximal_backdoor_extended} and \cref{prop:proximal_IPW_formula_extended} do not identify the full interventional distribution, we may choose to exclude certain observed proxies from appearing in the district factorization. To achieve this, we have to additionally project out these observed proxies in the latent projection step. We continue the discussion of this approach in \cref{sec:generalized_proximal_ID_algorithms}.

\subsubsection{Kernel operations}
\label{sec:Kernel_operations}
In order to identify all target kernels $\dodistr{D}{}{V^* \setminus D}$ starting from the observed joint distribution $p(V)$, we have to define kernel operations that act on a set of kernels and, under certain assumptions, output a different kernel. In this sense, kernel operations are the basic ``moves'' of any identification algorithm. 

We recall the definition of a conditional ADMG (CADMG) in \cref{sec:CADMGs_and_fixing} and associate to each kernel of the form $\dodistr{R}{}{S}$ the corresponding CADMG $\mathcal{G}(R,S)$ deduced from the causal graph $\mathcal{G}(V \cup U)$. We construct $\mathcal{G}(R,S)$ by latent projecting out all variables not contained in $R \cup S$ and removing all incoming edges to $S$, as detailed in \cref{sec:CADMGs_and_fixing}.

The most basic kernel operation, \texttt{Fix}, is a kernel generalization of both the backdoor formula \cref{eq:backdoor_formula} and marginalization. The necessary condition for \texttt{Fix} is a simple graphical criterion on the corresponding CADMG that we recall in \cref{sec:CADMGs_and_fixing}. If the condition applies to a variable $B$ in $\mathcal{G}(R,S)$ we call $B$ \textit{fixable}.

\begin{kernelbox}{Fixing  (\texttt{Fix}(B))}{}
\begin{tabular}{@{}ll}
\textbf{Input:} & $\dodistr{R}{}{S}$\\
\textbf{Output:} & $\dodistr{R \setminus \{B\}}{}{S \cup \{B\}}$ \\ 
\end{tabular} 

\vspace{0.5em}
\begin{tabular}{@{}ll}
\textbf{Conditions:}  $B \in R$ and $B$ is fixable in $\mathcal{G}(R,S)$.
\end{tabular} 
\end{kernelbox}

Next, we introduce kernel operations that generalize the proximal identification results from \cref{sec:outcome_and_treatment_bridge_functions} to the kernel setting. We separately define three kernel operations: \texttt{Obf} uses outcome bridge functions as in \cref{prop:proximal_backdoor_extended}, \texttt{Tbf} uses treatment bridge functions as in \cref{prop:proximal_IPW_formula_extended}, and \texttt{Ebf} uses the extended bridge functions from \cref{sec:extended_bridge_functions}. These kernel operations differ in their inputs, outputs, and necessary conditions to be applied.

\begin{kernelbox}{Outcome bridge function PCI method ($\texttt{Obf}_{W,Z}$(B))}{}
\begin{tabular}{@{}ll}
\textbf{Input:} & $\dodistr{R_1}{}{S}$ and $\dodistr{R_2}{}{S}$ \\
\textbf{Output:} & $\dodistr{R_1 \setminus (W \cup Z^* \cup\{B\})}{}{S \cup \{B\}}$ for $Z^* =  \de[\mathcal{G}(R_1 \setminus W,S)](B) \cap Z$
\end{tabular} 

\vspace{0.5em}
\textbf{Conditions:} 
\begin{enumerate}[itemsep=0.em, topsep=0.2em]
    \item $B \in R_1$ and $Z \subseteq R_1$,
    \item either
    $ \quad 
        W \subseteq R_1
        \quad \text{or} \quad
         W \cup (R_1 \setminus O) \subseteq R_2,
    $ \\
    where
    $ O := \de[\mathcal{G}(R_1 \setminus W,S)](B) \setminus (Z \cup \{B\})$,
    \item  there exist  $U^* \subseteq U$ such that \cref{app:ignorability_given_Z,app:kernel_valid_outcome_inducing_proxy,app:kernel_valid_treatment_inducing_proxy,app:outcome_bridge_function,app:outcome_completeness} hold for $R_1, R_2,B,W,Z,U^*$.    

\end{enumerate}
\end{kernelbox}

\begin{kernelbox}{Treatment bridge function PCI method ($\texttt{Tbf}_{W,Z}$(B))}{}
\begin{tabular}{@{}ll}
\textbf{Input:} & $\dodistr{R}{}{S}$ \\
\textbf{Output:} & $\dodistr{R \setminus (Z \cup W^* \cup \{B\})}{}{S \cup \{B\}}$ for $W^* =  \de[\mathcal{G}(R \setminus Z,S)](B) \cap W$
\end{tabular} 

\vspace{0.5em}
\textbf{Conditions:} 
\begin{enumerate}[itemsep=0.em, topsep=0.2em]
    \item $B \in R$ and $W,Z \subseteq R$,
    \item there exist  $U^* \subseteq U$ such that \cref{app:ignorability_given_W,app:kernel_valid_outcome_inducing_proxy,app:treatment_completeness,app:treatment_bridge_function,app:valid_treatment_inducing_proxy_with_W} hold for $R, B,W,Z,U^*$.    
\end{enumerate}
\end{kernelbox}

\begin{kernelbox}{Extended bridge function PCI method ($\texttt{Ebf}_{W,Z}$(B))}{}
\begin{tabular}{@{}ll}
\textbf{Input:} & $\dodistr{R}{}{S}$ \\
\textbf{Output:} & $\dodistr{R \setminus (Z^* \cup W^* \cup \{B\})}{}{S \cup \{B\}}$ \\
&for $W^* =  \de[\mathcal{G}(R,S)](B) \cap W$ and $Z^* =  \de[\mathcal{G}(R,S)](B) \cap Z$
\end{tabular} 

\vspace{0.5em}
\textbf{Conditions:} 
\begin{enumerate}[itemsep=0.em, topsep=0.2em]
    \item $B \in R$ and $W,Z \subseteq R$,
    \item there exist  $U^* \subseteq U$ such that   \cref{app:ignorability_given_W_Z,app:kernel_valid_outcome_inducing_proxy,app:valid_treatment_inducing_proxy_with_W} hold for $R, B,W,Z,U^*$,
    \item either \cref{app:extended_outcome_bridge_function,app:outcome_completeness} or alternatively  \cref{app:extended_treatment_bridge_function,app:treatment_completeness} hold.
\end{enumerate}
\end{kernelbox}

Comparing $\texttt{Obf}$ and \cref{prop:proximal_backdoor_extended}, we notice the additional exclusion of treatment-inducing proxies $Z^*$ from the identified kernel. \texttt{Obf} can use treatment-inducing proxies $Z^*$ that are causally affected by $B$ for setting up the bridge equation but cannot identify the causal effect that $B$ has on $Z^*$ since the this effect is in general confounded by $U^*$. This forces us to exclude $Z^*$ from the output kernel of $\texttt{Obf}$. An analogous exclusion was not necessary in \cref{prop:proximal_backdoor_extended}, since we assumed the exclusion restriction $Z(a) = Z$. Similarly, we can explain the asymmetries between \texttt{Tbf} and \cref{prop:proximal_IPW_formula_extended} and between \texttt{Ebf} and the results in \cref{sec:extended_bridge_functions}. The validity of \texttt{Obf}, \texttt{Tbf}, and \texttt{Ebf} can be shown very similarly to the proofs in \cref{sec:outcome_and_treatment_bridge_functions}. We provide proof sketches in \cref{sec:proof_Obf_validity,sec:proof_Tbf_validity,sec:proof_Ebf_validity}. 

Finally, we define \texttt{Cut} as the marginalization kernel operation that adds a variable $B$ to the intervention set by removing all variables from the kernel that are causally affected by $B$. 
\begin{kernelbox}{Cutting (\texttt{Cut}($B$))}
\begin{tabular}{@{}ll}
\textbf{Input:} & $\dodistr{R}{}{S}$ \\
\textbf{Output:} & $\dodistr{R \setminus \de[\mathcal{G}((V\setminus S),S)](B)}{}{S \cup \{B\}}$
\end{tabular}

\vspace{0.5em}
\textbf{Conditions:} None.
\end{kernelbox}

\subsubsection{Generalized proximal ID algorithms}
\label{sec:generalized_proximal_ID_algorithms}
To define an identification algorithm, we first choose a latent projection $\mathcal{G}(V \setminus H)$ as described in \cref{sec:district_factorization} and determine the corresponding district factorization by \cref{lem:district_factorization}. Starting from the observational distribution $p(V)$, we apply a sequence of kernel operations and obtain a corresponding sequence of identified kernels. For each target kernel $\dodistr{D}{}{V^*\setminus D}$ appearing in the district factorization, our objective is to construct a valid sequence of kernel operations that identifies $\dodistr{D}{}{V^* \setminus D}$. We refer to $\dodistr{D}{}{V^* \setminus D}$ as the \textit{target kernels} of our algorithm, which are defined by a choice of latent projection set $H$.

Using \texttt{Obf} and \texttt{Tbf} complicates the algorithm in two ways. First, we note that the kernel operation \texttt{Obf} takes two kernels as input and requires sufficient overlap between $R_1$ and $R_2$. Recall that the outcome bridge equation \eqref{eq:outcome_bridge_function} and the identification equation \eqref{eq:outcome_bridge_identifying_functional_extended} do not use the full $p(y,w,z,a,x)$ distribution, but instead use only the marginals $p(w,z,a,x)$ and $p(y,z,a,x)$. For this reason, \texttt{Obf} does not require a single kernel as input containing both $W$ and $Z$, provided that the intersection between $R_1$ and $R_2$ contains the analogs of $Z,A,X$. To take advantage of this additional flexibility, we apply kernel operations on two sequences of kernels $P_1$ and $P_2$ simultaneously and use them as separate inputs to \texttt{Obf}. This approach of using two kernels simultaneously was first described in \citet{Shpitser2021}. Unlike the outcome bridge approach, the treatment bridge approach does explicitly depend on the full $p(y,w,z,a,x)$ distribution as input. Therefore, $\texttt{Tbf}$ and similarly $\texttt{Ebf}$ require as input a single kernel that contains all relevant variables, and an analogous strategy of splitting the input between two kernels is not possible.

Second, both $\texttt{Obf}$ and $\texttt{Tbf}$ remove sets of observed proxy variables from $R_1$ without adding them to the intervention set, effectively removing them from the kernel. For a general kernel $\dodistr{A}{B}{C}$, we define $\texttt{Var}(\dodistr{A}{B}{C}) = A \cup B \cup C$ as the set of all variables contained in the kernel. For example, any kernel that we can identify after applying $\texttt{Obf}_{W,Z}$ will not depend on $W$, so for any kernel $P$, $\texttt{Var}(\texttt{Obf}_{W,Z}(B)[P]) \subsetneq \texttt{Var}(P)$. Therefore, we can only apply $\texttt{Obf}_{W,Z}$ to identify a target kernel $\dodistr{D}{}{V^* \setminus D}$ that does not depend on $W$, i.e. if $W \cap \texttt{Var}(\dodistr{D}{}{V^* \setminus D}) = \emptyset$ with $\texttt{Var}(\dodistr{D}{}{V^* \setminus D}) = V \setminus H$. Consequently, we might include in the latent projection set $H$ any proxy variable $W$ that we want to use in a kernel operation $\texttt{Obf}_{W,Z}$ as part of the algorithm. This ensures that $W$ does not appear in any of the target kernels. We can argue similarly for $\texttt{Tbf}_{W,Z}$ and the dependence of a target kernel on $Z$.
Thus, we could choose $H =  U \cup \setRV{W} \cup \setRV{Z}$, where $\setRV{W}$ contains all outcome-inducing proxies that are used in $\texttt{Obf}_{W,Z}$ kernel operations and $\setRV{Z}$ contains all treatment-inducing proxies that are used in $\texttt{Tbf}_{W,Z}$ kernel operations. This choice, however, is not necessarily optimal.\footnote{For some smaller choice of $H$, a given proxy variable $W$ might only appear in a target kernel that can be identified without ever needing to apply $\texttt{Obf}_{W,Z}$. In certain scenarios, we might thus want to choose a smaller $H \subsetneq U \cup \setRV{W} \cup \setRV{Z}$. \\
$\texttt{Obf}_{W,Z}(B)$ also removes a subset of $Z$ that is causally affected by $B$ from the kernel but not contained in $\setRV{W} \cup \setRV{Z}$. In other scenarios, we might thus also want to choose a larger $H \supsetneq U \cup \setRV{W} \cup \setRV{Z}$. If we can exclusively use \texttt{Ebf} and \texttt{Fix}, we might not have to include any observed proxies in $H$ at all.} For this reason, we formulate the generalized proximal ID algorithm in a way that leaves the choice of the latent projection set $H$ open.

At first glance, \texttt{Ebf} might appear to be strictly stronger compared to \texttt{Obf} and  \texttt{Tbf}. Although \texttt{Ebf} can, in principle, identify a larger marginal of the interventional distribution, it relies on stricter assumptions than \texttt{Obf} and \texttt{Tbf}. Furthermore, \texttt{Obf} offers the additional flexibility to take two kernels as input. We will therefore treat all three proximal kernel operations as different tools available to the algorithm, requiring different inputs and different assumptions, and yielding different output kernels. This yields the following general framework for proximal ID algorithms. 

\begin{theorem}[Generalized proximal ID algorithm] 
    \label{thm:generalized_proximal_ID_algorithm}
    Let a causal graph $\mathcal{G}(V \cup U)$ with disjoint sets $A, Y \subseteq V$ be given. Choose a latent projection set $H \supseteq U$. Let $V^* = V \setminus H$ and $Y^*$ be the set of ancestors of $Y$ in $\mathcal{G}(V^*)$ via directed paths that do not intersect $A$. Then $\dodistr{Y}{}{A}$ is identified from $p(V)$ in the causal model represented by $\mathcal{G}(V \cup U)$ given latent projection set $H$, if for every district $D \in \mathcal{D}(\mathcal{G}(V^*)_{Y^*})$ there exists an enumeration $B_1, \dots ,B_n$  of $(V^* \setminus D)$ and tuples
    \begin{align*}
        (B_1,K_1,W_1,Z_1), \dots, (B_n,K_n,W_n,Z_n) \quad \text{with }  K_i \in \{\texttt{Fix}, \texttt{Obf}, \texttt{Tbf}, \texttt{Ebf}\} \text{ and } W_i, Z_i \subseteq V \setminus (A \cup Y) \; ,
    \end{align*}
    for $i=1, \dots, n$, such that \cref{alg:proximal_ID_step_1} terminates successfully. Then we have
    \begin{align*}
        \dodistr{Y}{}{a} = \sum_{Y^* \setminus Y} \prod_{D \in \mathcal{D}(\mathcal{G}(V^*)_{Y^*})} \dodistr{D}{}{s_D},
    \end{align*}
    with every kernel $\dodistr{D}{}{s_D}$ identified via \cref{alg:proximal_ID_step_1} with a suitable sequence of tuples $(B_i,K_i,W_i,Z_i)$.

\begin{algorithm}[h!]
\caption{Proximal identification step}
\label{alg:proximal_ID_step_1}
\begin{algorithmic}[1]

\State \textbf{Input:} causal model with hidden variable DAG $\mathcal{G}(V\cup U)$,  and
tuples 
\State $ (B_1,K_1,W_1,Z_1), \dots, (B_n,K_n,W_n,Z_n) \in V \times \{\texttt{Fix}, \texttt{Obf}, \texttt{Tbf}, \texttt{Ebf}\} \times \mathcal{P}(V) \times \mathcal{P}(V)$.
\State \textbf{Output:} $\dodistr{D}{}{V^* \setminus D}$ or FAIL.

\vspace{0.3em}

\State $P_1 \gets p(V)$ 
\State $P_2 \gets p(V)$

\vspace{0.3em}

\For{$i = 1, \dots ,n$}
    \State $(B,K,W,Z) \gets (B_i,K_i,W_i,Z_i)$
    \vspace{0.3em}
    \State \textbf{Step 1: Update $P_1$}
    \If{$K = \texttt{Fix}$ and \texttt{Fix}(B) conditions hold for $P_1$}
        \State $P_1 \gets \texttt{Fix}(B)[P_1]$
    \ElsIf{$K = \texttt{Obf}$ and $\texttt{Obf}_{W,Z}(B)$ conditions hold for $(P_1,P_2)$}
        \State $P_1 \gets \texttt{Obf}_{W,Z}(B)[P_1,P_2]$
    \ElsIf{$K = \texttt{Tbf}$ and $\texttt{Tbf}_{W,Z}(B)$ conditions hold for $P_1$}
        \State $P_1 \gets \texttt{Tbf}_{W,Z}(B)[P_1]$
    \ElsIf{$K = \texttt{Ebf}$ and $\texttt{Ebf}_{W,Z}(B)$ conditions hold for $P_1$}
        \State $P_1 \gets \texttt{Ebf}_{W,Z}(B)[P_1]$
    \Else
        \State \Return FAIL
    \EndIf
    \vspace{0.3em}
    \State \textbf{Step 2: Update $P_2$}
    \If{\texttt{Fix}(B) condition holds for $P_2$}
        \State $P_2 \gets \texttt{Fix}(B)[P_2]$
    \Else
        \State $P_2 \gets \texttt{Cut}(B)[P_2]$    
    \EndIf
    \vspace{0.3em}
    \State \textbf{Step 3: Check kernel}
    \If{$V^* \not\subseteq \texttt{Var}(P_1)$}
        \State \Return FAIL  
    \EndIf
\EndFor

\vspace{0.3em}

\State \Return marginal kernel of $P_1$ on $D$

\end{algorithmic}
\end{algorithm}

\end{theorem}

A proof sketch of \cref{thm:generalized_proximal_ID_algorithm} is given in \cref{sec:proof_of_generalized_ID_algorithm}. The original proximal ID algorithm proposed by \citet{Shpitser2021} can be interpreted as a special case of \cref{thm:generalized_proximal_ID_algorithm} by imposing the conditions $W_i \subseteq H$ and $K_i \in \{\texttt{Fix},\texttt{Obf}\}$. Similarly, the standard ID algorithm via fixing from \citet{Richardson2023} is a special case with $H = U$ and all $K_i = \texttt{Fix}$. For notational clarity, the algorithm presented here does not keep track of past kernel operations and only outputs the identified kernel. An extension to output the identifying functional itself is straightforward. 

An important aspect of the algorithm is the introduction of a second sequence of kernels, denoted by $P_2$, that is used in \texttt{Obf}. At each step, we add the same variable $B$ to the intervention set in both kernels $P_1$ and $P_2$ using different kernel operations. $P_1$ might remove a proxy $W$ or $Z$ from its kernel sequence by applying $\texttt{Obf}$ or $\texttt{Tbf}$, while $P_2$ still retains the proxy variable. Therefore, $P_2$ ``aims'' to keep proxies included in the problem. This allows us to possibly reuse certain proxies multiple times throughout the identification algorithm, as the required kernels that include these proxies are provided by $P_2$.

Instead of using the kernels of the form $\dodistr{D}{}{V^* \setminus D}$ as target kernels, we could choose to target the kernels $\dodistr{D}{}{s_D}$ from the district factorization directly. As we argued above, the identification of $\dodistr{D}{}{s_D}$ is generally a weaker condition compared to the identification of $\dodistr{D}{}{V^*\setminus D}$.\footnote{In the standard non-parametric identification problem, \cite{Shpitser2006} shows that the identifiability of both kernels is equivalent. However, this equivalence will, in general, not be true for the proximal identification problem.} In this setting, we only require a sequence of $\{(B_i,K_i,W_i,Z_i)\}_{i=1}^m$ with $B_1, \dots, B_m$ enumerating $\pa(D) \setminus D$ and similarly need to enforce $\pa(D) \cup D \subseteq \texttt{Var}(P_1)$ at each step of the algorithm. In accordance with the proximal ID algorithm in \citet{Shpitser2021}, we  focus on the kernels of the form $\dodistr{D}{}{V^* \setminus D}$ as target kernels and do not further explore other choices of target kernels at this point.

The proximal ID algorithm yields only a sufficient condition for identifiability and is therefore sound but not complete. In contrast, the standard ID algorithm provides a sound and complete solution to \cref{prob:nonparametric_identification_problem}. More precisely, for any interventional distribution that the ID algorithm fails to identify, there exist two causal models on the given graph that agree on the observational distribution but disagree on the interventional distribution of interest. Hence, the ID algorithm also provides a necessary criterion for nonparametric identifiability in \cref{prob:nonparametric_identification_problem}. An analogous completeness result is not available for the proximal ID algorithm in \cref{prob:nonparametric_proximal_identification_problem}.

\subsubsection{Applying the generalized proximal ID algorithm}
\label{sec:applying_the_generalized_proximal_ID_algorithm}
We apply the generalized proximal ID algorithm described in \cref{thm:generalized_proximal_ID_algorithm} to the proximal front-door identification problem discussed in \cref{Subsec:Proximal_front_door}, based on the front-door graph in \cref{fig:proximal_front_door_graph}. 

One can easily check that $\dodistr{Y}{}{a}$ cannot be identified using only $\texttt{Fix}$.  Thus, we need to apply either $\texttt{Obf}$, $\texttt{Tbf}$, or $\texttt{Ebf}$. We can only use $W$ as outcome-inducing proxy and $Z$ as treatment-inducing proxy. As in \cref{thm:proximal_front_door_treatment_bridge}, we may choose to apply \texttt{Tbf} and use a treatment bridge function approach. To this end, let $H = \{U\}$, leading with $Y^*=\{Y,W,X,M\}$ to the districts $D_1 = \{M\}$ and $D_2 = \{Y,W,X\}$ and the target kernels $\dodistr{M}{}{A,W,X}$ and $\dodistr{Y,W,X}{}{A,M}$. We can easily identify $\dodistr{M}{}{A,W,X}$ using standard fixing operations. To identify $\dodistr{Y,W,X}{}{A,M}$, we first apply $\texttt{Fix}(M)$ to $p(Y,A,M,W,Z,X)$, yielding $\dodistr{Y,A,W,Z,X}{}{M}$. Then we can apply $\texttt{Tbf}_{W,Z}(A)$ to finally obtain $\dodistr{Y,W,X}{}{A,M}$. We observe that the proximal ID algorithm proceeds through the same sequence of steps as the proof of \cref{thm:proximal_front_door_treatment_bridge}. In particular, we recognize \cref{ass:kernel_district_decomposition} and \cref{ass:kernel_exclusion_restrictions_M} as algebraic conditions for the district factorization and the fixing ignorability to identify $\dodistr{M}{}{A,W,X}$. Similarly, \cref{ass:kernel_ignorability_M} is the necessary ignorability condition for $\texttt{Fix}(M)$, and \cref{ass:kernel_latent_ignorability,ass:kernel_valid_outcome_inducing_proxy,ass:kernel_valid_treatment_inducing_proxy,ass:kernel_treatment_completeness,ass:kernel_treatment_bridge_function} ensure that we can subsequently apply $\texttt{Tbf}_{W,Z}(A)$. 

As an alternative approach, we could try to apply an extended bridge function operation. However, $Z$ is a post-treatment proxy and is therefore included
in the set $Z^*$ of treatment proxies that are causally affected by $A$. Consequently, \texttt{Ebf} cannot recover an interventional kernel involving
$Z$. Thus, in this case, applying \texttt{Ebf} offers no further identification results beyond \texttt{Tbf}.

If we instead try to apply the proximal ID algorithm from \citet{Shpitser2021} to this graph, we are restricted to using only $\texttt{Fix}$ and $\texttt{Obf}$ and have to choose $H = \{U,W\}$. With $Y^*=\{Y,M,X\}$, we obtain the single district $\{Y,M,X\}$. This requires us to identify the single target kernel $\dodistr{Y,M,X}{}{A}$. But we are unable to apply any suitable kernel operation to $p(Y,A,M,W,Z,X)$, as neither the conditions for $\texttt{Fix}(A)$ nor for $\texttt{Obf}_{W,Z}(A)$ are satisfied. The proximal front-door graph in \cref{fig:proximal_front_door_graph} therefore provides an example in which \cref{thm:generalized_proximal_ID_algorithm} succeeds at identifying a kernel that remains unidentified by the algorithm described in \citet{Shpitser2021}.

The proximal front-door setting can also easily be modified to yield a graph in which the identification of the full joint interventional distribution via \texttt{Ebf} becomes necessary. The causal graph in \cref{fig:proximal_front_door_graph_comparison}(e) is a combination of both proximal front-door graphs (a) and (b) in \cref{fig:proximal_front_door_graph_comparison} with both $W$ and $Z$ affecting the mediator $M$. For \texttt{Obf} , we require $H = \{U,W\}$, and for \texttt{Tbf} , we require $H = \{U,Z\}$, leading to a single district containing both $M$ and $Y$ in both cases. But to apply any proximal identification method here, we need to fix $M$ first. Thus, neither \texttt{Obf} nor \texttt{Tbf} alone can suffice to identify $p(Y(a))$. If we are instead able to use \texttt{Ebf}, we may choose $H = \{U\}$. With $Y^* = \{Y,M,W,Z,X\}$, this induces the districts $D_1 = \{M\}$ and $D_2 = \{Y,W,Z,X\}$. We identify the target kernel of $D_1$ trivially via fixing operations and the target kernel of $D_2$ via $\texttt{Fix}(M)$ and $\texttt{Ebf}_{W,Z}(A)$. 

Overall, any proximal ID algorithm depends strongly on a ``good'' choice of kernel operations and proxy sets. This adds substantial combinatorial complexity compared to the standard ID algorithm. \citet{Richardson2023} showed that applying $\textsc{Fix}(B)$ does not change whether a different vertex $B'$ is fixable in the output kernel. Therefore, the \textsc{Fix} operations commute, eliminating the need for backtracking in the standard ID algorithm and reducing its computational complexity down to a low order polynomial time \citep{Shpitser2006,Shpitser2021}. The proximal ID algorithm introduces additional combinatorial complexity through the choice of latent projection set $H$ as well as in the choices of $(B_i,K_i,W_i,Z_i)$ at each step. From the front-door example, we could formulate the following heuristic: select proxy sets and subsequently $H$ such that the size of the districts in $\mathcal{G}(V \setminus H)$ is minimized. The problem of how to reduce the combinatorial complexity in a rigorously justifiable manner remains unresolved.

Similar to \citet{Shpitser2021}, we may argue that in any real-world scenario, the choices of proxy variables must be made very carefully by a human domain expert and should be strongly backed by real-world data. These choices cannot reasonably be made algorithmically without human insight. Consequently, proximal ID algorithms are best understood as a tool for assessing whether proximal causal inference methods are, in principle, sufficient for a given identification problem, given a certain set of assumptions, rather than for directly suggesting a practically applicable solution.

\newpage
\section{Summary}
\label{sec:summary}
We provided an overview of different identification methods for interventional kernels in NPSEM-IE causal models using proximal causal inference methods. We reviewed the main identification results from \citet{Shpitser2021} using outcome bridge functions and established similar results using the treatment bridge functions introduced in \citet{Cui2023}. Furthermore, we investigated how tools from proximal causal inference can be applied in a more general setting, complementing standard nonparametric identification theory. 

Our primary contribution is the introduction of extended bridge functions, together with two novel identification results for interventional kernels containing both outcome-inducing and treatment-inducing proxy variables. The second main contribution is a generalization of the proximal ID algorithm from \citet{Shpitser2021} within a framework of kernel operations, using both outcome and treatment bridge functions, as well as the newly introduced extended bridge functions. We applied this algorithm to multiple variations of the proximal front-door graph and derived a new identification result in this setting, complementing the proximal front-door criterion of \citet{Shpitser2021}. 

Many directions can be explored in future work. First, the extended bridge functions from \cref{sec:extended_bridge_functions} can be further investigated in the more general setting of continuous state spaces and general kernels. This will require refined regularity conditions to ensure the existence of solutions to the extended bridge equations. We provided an example application in which extended bridge functions confer advantages over standard outcome and treatment bridge functions. Many further applications remain to be explored. Extended bridge functions could, for example, be applied to the setting of optimal individualized treatment regimes, extending the results in \citet{Shen2023}.

Second, the generalized framework for proximal ID algorithms could be applied to other settings, including proximal mediation analysis \citep{Dukes2023, Ghassami2024} and proximal g-computation \citep{TchetgenTchetgen2024}. These settings provide additional applications in which proximal ID algorithms may yield new identification results or new identifying functionals under different conditions. 

Third, several technical extensions of the generalized proximal ID algorithm follow directly from the proposed framework. In general, the challenge of proximal identification lies in choosing which kernels to target and then exploring possible sequences of kernel operations and associated proxy sets in order to identify the target kernels. We presented an algorithm that uses district factorization to define its target kernels, relies on a specific selection of kernel operations, and keeps track of two sequences of kernels at the same time. These design choices for the algorithm are not necessarily optimal. While district factorization yields a necessary and sufficient condition for identification in the standard nonparametric setting, it only provides a sufficient condition for proximal identification. Furthermore, the implementation of kernel operations that take multiple kernels as inputs could be expanded by adding additional sequences of kernels to the algorithm. Overall, this framework of proximal ID algorithms can be regarded as a search algorithm operating in a space of kernels. The question of optimal search strategies remains unanswered and could be explored further.

Finally, we have not addressed statistical estimation. While much work exists on estimating solutions to integral equations in a variety of different settings, for example \citet{Kallus2021,Bennett2025,Cui2023}, the specific estimation problems arising in proximal ID algorithms have not been studied in detail. The statistical properties of an identifying functional will, in general, depend on the kernel operations used for identification. Selecting among different kernel operation sequences and constructing a practical estimator for the identified quantities gives rise to many nontrivial problems that remain open for future work.
\newpage 

\section*{Acknowledgements}

I am grateful to Qingyuan Zhao for interesting discussions and guidance throughout this project. I also thank Eric Tchetgen Tchetgen for comments and suggestions and Ilya Shpitser, Wang Miao and Axel Munk for useful discussions.

The author would like to thank the Isaac Newton Institute for Mathematical Sciences, Cambridge, for support and hospitality during the programme Causal inference: From theory to practice and back again, where work on this paper was undertaken. This work was supported by EPSRC grant EP/Z000580/1.

\phantomsection
\pdfbookmark[1]{Appendix}{pdfbk:appendix}
\section*{Appendix}
\label{sec:appendix}
\renewcommand{\thesubsection}{\Roman{subsection}}
\crefalias{section}{appsec}
\crefalias{subsection}{appsubsec}

\numberwithin{theorem}{subsection}
\setcounter{theorem}{0}

\aliascntresetthe{lemma}
\aliascntresetthe{corollary}
\aliascntresetthe{proposition}
\aliascntresetthe{definition}
\aliascntresetthe{remark}
\aliascntresetthe{example}

\subsection{Graph terminology}
\label{sec:Graphs_and_DAGs}
\begin{definition}
In a (mixed) graph $\mathcal{G} = (V,E)$ with vertices $v,w \in V$, we denote a directed edge between two vertices as $v \rightarrow w$ and a bidirected edge as $v \leftrightarrow w$. 
To describe the relationship between vertices, we use the terminology of kinship and call
\begin{enumerate}[i.)]
    \item the vertex $v$ a \emph{parent} of a vertex $w$, and $w$ a \emph{child} of $v$ if $v \rightarrow w$, 
    \item the vertex $w$ a \emph{descendant} of a vertex $v$ if there exists a directed path from $v$ to $w$. Conversely, we call $v$ an \emph{ancestor} of $w$.
\end{enumerate}
We define $\pa(W)$, $\an(W)$, and $\de(W)$ as the set containing all vertices that are parents, ancestors, and descendants of a vertex $w \in W$, respectively. In particular, $\de(W)$ contains $W$ itself.
\end{definition}

\subsection{Sufficient conditions for the existence of bridge functions}
\label{sec:conditions_for_existence_of_bridge_functions}
We recall from \citet{Miao2018} and \citet{Cui2023} the sufficient conditions for the existence of bridge functions $h$ and $q$ solving the bridge equations \eqref{eq:outcome_bridge_function} and \eqref{eq:treatment_bridge_function}. The following conditions are reproduced from Appendix B in \citet{Cui2023}.

For any cumulative distribution function $F: \mathbb{R}^n \rightarrow\mathbb{R^+}$, we define $L^2\{F (t)\}$ as the space of all functions mapping from $\mathbb{R}^n$ to $\mathbb{R}$ that are square integrable functions with respect to the measure $\mu_F$ induced by $F$ on $\mathbb{R}^n$. Then $L^2\{F (t)\}$ is a Hilbert space with inner product $\langle g_1, g_2\rangle = \int g_1(t)g_2(t)\diff F (t)$. We define two conditional expectation operators
\begin{align*}
    T_{a,x}: \;&L^2\{F(w|a, x)\} \rightarrow L^2\{F (z|a, x)\}, \; h \mapsto \mathbb{E}[h(W ) \, | \, z, a, x], \\
    T'_{a,x}:\;&L^2\{F(z|a, x)\} \rightarrow L^2\{F (w|a, x)\}, \; q \mapsto \mathbb{E}[q(Z) \, | \, w, a, x].
\end{align*}
Let $(\lambda_{a,x,n}, \phi_{a,x,n}, \psi_{a,x,n})_{n=1}^\infty$ be a singular value decomposition of $T_{a,x}$, and $(\lambda'_{a,x,n}, \phi'_{a,x,n}, \psi'_{a,x,n})_{n=1}^\infty$ be a singular value decomposition of $T'_{a,x}$. Now we can state \cref{auxass:outcome_bridge,auxass:treatment_bridge} that guarantee the existence of solutions to the bridge equations \eqref{eq:outcome_bridge_function} and \eqref{eq:treatment_bridge_function}, respectively.
\begin{auxassumption}
    \label{auxass:outcome_bridge}
    We assume completeness of $W$ with respect to $Z$, in the sense that for every $g \in L^2(Z)$ and for all $a \in \mathcal{X}_A$, $x \in \mathcal{X}_X$ 
    \begin{align*}
        \mathbb{E}[g(Z)\,| \,W, A = a, X = x] = 0 \text{ almost surely} \ \Longleftrightarrow \ g(Z) = 0 \text{ almost surely.}
    \end{align*}
    Furthermore, we assume the following regularity condition
    \begin{enumerate}
        \item $\int \int f(w \, | \, z,a,x) \, f(z \, | w,a,x) \diff w \diff z < \infty$,
        \item $\int f^2(y \, | \, z,a,x)  f(z \,| \, a,x) \diff z < \infty$,
        \item $\sum_{n=1}^\infty \lambda_{a,x,n}^{-2} |\langle f(y \, | \, z,a,x) , \phi_{a,x,n} \rangle|^2 < \infty$.
    \end{enumerate}
\end{auxassumption}

\begin{auxassumption}
    \label{auxass:treatment_bridge}
     We assume completeness of $Z$ with respect to $W$, in the sense that for every $g \in L^2(W)$ and for all $a \in \mathcal{X}_A$, $x \in \mathcal{X}_X$ 
    \begin{align*}
        \mathbb{E}[g(W)\,| \,Z, A = a, X = x] = 0 \text{ almost surely} \ \Longleftrightarrow \ g(W) = 0 \text{ almost surely.}
    \end{align*}
    Furthermore, we assume the following regularity conditions:
    \begin{enumerate}
        \item $\int \int f(w \, | \, z,a,x) \, f(z \, | w,a,x) \diff w \diff z < \infty$,
        \item $\int f^{-2}(a \,| \, w,x)  f(w \,| \, a,x) \diff w < \infty$,
        \item $\sum_{n=1}^\infty \lambda_{a,x,n}^{'-2} |\langle f^{-1}(a \,| \, w,x), \phi'_{a,x,n} \rangle|^2 < \infty$.
    \end{enumerate}
\end{auxassumption}

\subsection{Extended bridge functions for general state spaces}
\label{sec:extended_bridge_functions_for_general_state_spaces}
We give a proof sketch of \cref{cor:identify_entire_joint_with_outcome_bridge} and \cref{cor:identify_entire_joint_with_treatment_bridge} for general, non-discrete random variables under the additional assumption that the integral equations \eqref{eq:extended_outcome_bridge} and \eqref{eq:extended_treatment_bridge} admit solutions in this general setting. As noted above, the practical significance of these results critically depends on the extent to which the corresponding integral equations can be solved in a more general framework that extends beyond the discrete setting

\begin{corollary}
    \label{app:identify_entire_joint_with_outcome_bridge}
    Let $A,Y,W,Z,X$ and $U$ be (not necessarily discrete) random variables. If we assume the existence of an extended outcome bridge function $h$ in \cref{ass:extended_outcome_bridge_function}, the joint distribution $\dodistr{Y,W,Z,X}{}{a}$ is identifiable under \cref{ass:ignorability_given_W_Z,ass:exclusion_restriction_W_Z_X,ass:valid_treatment_inducing_proxy,ass:valid_outcome_inducing_proxy,ass:valid_treatment_inducing_proxy_with_W,ass:outcome_completeness}.
    \begin{proof}
        As detailed in the proof of \cref{cor:identify_entire_joint_with_outcome_bridge},        \cref{ass:valid_outcome_inducing_proxy,ass:valid_treatment_inducing_proxy_with_W,ass:ignorability_given_W_Z} imply conditional independencies \eqref{eq:valid_treatment_inducing_proxy_as_Y_prime} and \eqref{eq:ignorability_for_Y_prime}. We define the outcome variable $Y' = (Y,W)$. We can easily deduce the  assumptions that are necessary to apply \cref{prop:proximal_backdoor_extended} to $(Y',W,Z,A,X)$. In particular, we need analogs of \cref{ass:valid_outcome_inducing_proxy,ass:valid_treatment_inducing_proxy,ass:latent_positivity,ass:outcome_bridge_function,ass:outcome_completeness,ass:exclusion_restriction_Z_X,ass:ignorability_given_Z} for the new outcome $Y'$.
        \begin{enumerate}
            \item Equations \eqref{eq:ignorability_for_Y_prime} and \eqref{eq:valid_treatment_inducing_proxy_as_Y_prime} give versions of latent ignorability \cref{ass:ignorability_given_Z} and \cref{ass:valid_treatment_inducing_proxy} for the new outcome $Y'$.
            \item Assumption \ref{ass:valid_outcome_inducing_proxy} remains unchanged under the redefined outcome variable $Y'$.
            \item The positivity \cref{ass:ignorability_given_Z} remains unchanged and is implied by \cref{ass:ignorability_given_W_Z}.
            \item The exclusion restrictions in \cref{ass:exclusion_restriction_Z_X} remain unchanged and are implied by \cref{ass:exclusion_restriction_W_Z_X}.
            \item The completeness \cref{ass:outcome_completeness} also remains unchanged.
            \item Finally, the existence of a bridge function for $Y'$ corresponding to \cref{ass:outcome_bridge_function} is given in \cref{ass:extended_outcome_bridge_function}.
        \end{enumerate}
        Then we can apply the proximal backdoor formula from \cref{prop:proximal_backdoor_extended} to conclude
        \begin{align*}
            \dodistr{y',z,x}{}{a} = \dodistr{y,w,z,x}{}{a} = \sum_{w'} h(y,w,w',a,x) \, p(w',z,x).
        \end{align*}
        \end{proof}
\end{corollary}

\begin{corollary}
    \label{app:identify_entire_joint_with_treatment_bridge}
    Let $A,Y,W,Z,X$ and $U$ be (not necessarily discrete) random variables. If we assume the existence of an extended treatment bridge function $q$ in \cref{ass:extended_treatment_bridge_function}, the joint distribution $\dodistr{Y,W,Z,X}{}{a}$ is identifiable under \cref{ass:ignorability_given_W_Z,ass:exclusion_restriction_W_Z_X,ass:valid_outcome_inducing_proxy,ass:valid_treatment_inducing_proxy_with_W,ass:proxy_positivity,ass:treatment_completeness}.
    \begin{proof}
        As in \cref{cor:identify_entire_joint_with_treatment_bridge}, we use \cref{ass:valid_outcome_inducing_proxy} and the existence of a bridge function in \cref{ass:extended_treatment_bridge_function} to calculate
        \begin{align*}
            \sum_{u} \frac{p(z \, | \, u,x)}{p(a \, |\, u,x)} p(u \, | \, w,a,x) = \frac{p(z \, | \, w,x)}{p(a \, | \, w,x)} = \sum_u \sum_{z'} q(z,z',a,x) \, p(z' \, | \, u,a,x) \, p(u \, | \, w,a,x).
        \end{align*}
        With completeness \cref{ass:treatment_completeness}, we conclude
        \begin{align}
            \label{eq:app_bridge_equation_with_Z}
            \frac{p(z \, | \, u,x)}{p(a \, |\, u,x)} = \sum_{z'} q(z,z',a,x) \, p(z' \, | \, u,a,x).
        \end{align}
        Finally, we can use \cref{ass:ignorability_given_W_Z,ass:exclusion_restriction_W_Z_X} and consistency to conclude
        \begin{align*}
            \dodistr{y,w,z,x}{}{a} &= \sum_{u} p(y \, | \, a,w,z,u,x) \, p(w,z,u,x) \\
            &= \sum_{u} p(y \, | \, a,w,u,x) \,p(w,u,x)\, p(z \,| u,x) && \text{(by \ref{ass:valid_outcome_inducing_proxy} and \ref{ass:valid_treatment_inducing_proxy_with_W})} \\
            &= \sum_{u} p(y,a,w,u,x) \,\frac{p(z \, | \, u,x)}{p(a \, |\, u,x)}  && \text{(by \ref{ass:valid_outcome_inducing_proxy})}\\
            &= \sum_{u} p(y,a,w,u,x) \sum_{z'} q(z,z',a,x) \, p(z' \, | \, u,a,x) && \text{(by \ref{eq:app_bridge_equation_with_Z})} \\
            &= \sum_{z'} q(z,z',a,x) \, p(y,w,z',a,x).   && \text{(by  \ref{ass:valid_outcome_inducing_proxy} and \ref{ass:valid_treatment_inducing_proxy_with_W})} 
        \end{align*}
        \end{proof}
\end{corollary}

\subsection{ADMGs and latent projections}
The following definitions were adapted from \citet{Verma1990}, \citet{Richardson2023}, and \citet{Shpitser2021}.
\label{sec:ADMGs_and_latent_projections}
\begin{definition}[Acyclic Directed Mixed Graph (ADMG)] 
        A mixed graph $\mathcal{G}= (V,E)$ is called an \textit{Acyclic Directed Mixed Graph} (ADMG) if it contains no directed cycles and at most one directed edge and one bidirected edge between each pair of vertices. 
\end{definition}
\begin{definition} [Districts] Consider an ADMG $\mathcal{G}= (V,E)$.
    \begin{enumerate}[i.)]
    \item We call a set $W \subset V$ \emph{bidirectionally connected} if, for all $w_1, w_2 \in W$, there exists a bidirected path connecting $w_1$ and $w_2$.
    \item A bidirected connected set $D$ is called a \emph{district} if it is maximally bidirected connected, that is, there exists no larger bidirected connected set $W \supsetneq D$. We define $\mathcal{D}(\mathcal{G})$ as the set of all districts of the ADMG. For $v \in V$, we refer to the district containing $v$ as $\dis(v)$. 
    \end{enumerate}
\end{definition}

\begin{definition}[Latent projections]
    Consider a DAG or ADMG $\mathcal{G}(V) = (V,E)$ and a partition $V = A \cup B$. A \emph{latent projection} of $\mathcal{G}(V)$ onto $A$ yields an ADMG $\mathcal{G}(A)$ on the vertex set $A$ such that for all vertices $a_1, a_2 \in A$
    \begin{enumerate}
        \item $a_1 \rightarrow a_2$ in $\mathcal{G}(A)$ iff there exists a directed path from $a_1$ to $a_2$ in $\mathcal{G}(V)$ and all non-endpoint vertices are in $B$, for example $a_1 \rightarrow a_2$ or $a_1 \rightarrow b_1 \rightarrow \dots \rightarrow b_k \rightarrow a_2$ for $b_1, \dots, b_k \in B$.
        \item $a_1 \leftrightarrow a_2$ in $\mathcal{G}(A)$ iff there exists a path between $a_1$ and $a_2$ such that
        the non-endpoints are all non-colliders in $B$ and such that the edge adjacent to $a_1$
        and the edge adjacent to $a_2$ both have arrowheads at those vertices. For example, $a_1 \leftrightarrow b_1 \rightarrow \dots  \rightarrow b_k \rightarrow a_2$ for $b_1, \dots, b_k \in B$.
    \end{enumerate}
\end{definition}
As a direct consequence of the definition, we may conclude that latent projections compose trivially. For $T \subseteq U \subseteq V$, latent projection $\mathcal{G}(V)$ onto $\mathcal{G}(U)$ followed by a latent projection onto $\mathcal{G}(T)$ yields the same ADMG as directly performing a latent projection of $\mathcal{G}(V)$ onto $\mathcal{G}(T)$.

\subsection{CADMGs and fixing}
\label{sec:CADMGs_and_fixing}
Consider an ADMG $\mathcal{G}(V)$ that is obtained by latent projection of a hidden variable DAG $\mathcal{G}(V \cup U)$. Based on \citet{Richardson2023}, we define a CADMG $\mathcal{G}(R,S)$ as a graphical representation of the kernel $\dodistr{R}{}{S}$ for some $R,S \subset V$.

\begin{definition}
    We define the conditional ADMG (CADMG) corresponding to $\dodistr{R}{}{S}$ in two steps.
    \begin{enumerate}
        \item We perform a latent projection on the ADMG $\mathcal{G}(V)$ to yield the CADMG $\mathcal{G}(R \cup S)$, projecting out the vertices in $V \setminus (R \cup S)$.
        \item We obtain an ADMG $\mathcal{G}(R, S)$ on the vertex set $R \cup S$ with two disjoint sets of vertices $R$ and $S$ by removing all incoming edges of the form $v \rightarrow s$ or $v \leftrightarrow s$ for all $s \in S$, $v \in V$. This bipartite graph is called a CADMG.
    \end{enumerate}
\end{definition}

In the derivation of the proximal ID algorithm in \citet{Shpitser2021}, the order of performing the intervention operation on the vertex set $S$ and latent projection is swapped; i.e., we latent project a suitable CADMG. Our definition is equivalent and closer to \citet{Richardson2023}. 
\begin{definition}
    \label{def:fixable_vertex}
    Given a CADMG $\mathcal{G}(R, S)$, the set of fixable vertices is defined as 
    \begin{align*}
        \mathbb{F}(\mathcal{G}(R,S)) := \{ r \in R \, | \,  \dis[\mathcal{G}(R, S)](r) \cap \de[\mathcal{G}(R, S)](r) = \{r\} \}.
    \end{align*}
\end{definition}

\begin{proposition}
 \label{prop:fix_kernel_operation_valid}
 The \texttt{Fix} kernel operation defined in \cref{sec:Kernel_operations} is valid.
 \begin{proof}
    Consider a causal model with the hidden variable DAG $\mathcal{G}(V \cup U)$ as the causal graph. 
    Let $S,R \subseteq V$ be disjoint. \citet{Richardson2023} show that if $B \in R$ is fixable in $\mathcal{G}(R,S)$, then $\dodistr{R \setminus \{B\}}{}{S \cup \{B\}}$ is identified from $\dodistr{R}{}{S}$ via 
    \begin{align*}
        \dodistr{R \setminus \{B\}}{}{S \cup \{B\}} = \frac{\dodistr{R}{}{S}}{\dodistr{B}{\mb[\mathcal{G}(R,S)](B) }{S}},
    \end{align*}
    where we define the Markov blanket $\mb(B) = \pa[\mathcal{G}(R,S)](\dis[\mathcal{G}(R,S)](B)) \setminus (S \cup \{B\})$ as the set of all random vertices that are either in the same district as $B$ or a parent of the vertices in that district. 
 \end{proof}
\end{proposition}

\subsection{Outcome bridge function kernel operation \texttt{Obf}}
\label{sec:proof_Obf_validity}
We state the additional conditions of the kernel operation $\texttt{Obf}$ and prove the identification of the output kernel. Let $\dodistr{R_1}{}{S}$ and $\dodistr{R_2}{}{S}$ be two kernels, let $B \in R_1$ be the treatment variable, let $U^* \subseteq U$ be the latent confounders, and let $Z \subseteq R_1$ be the treatment-inducing proxy variables. In addition, let outcome-inducing proxy variables $W$ be given such that we have access to $\dodistr{W \cup (R_1 \setminus O)}{}{S}$ for ``outcome'' variables $O := \de[\mathcal{G}(R_1 \setminus W,S)](B) \setminus (Z \cup \{B\})$. That is, we require either $W \subseteq R_1$ or, alternatively, $W \cup (R_1 \setminus O) \subseteq R_2$.

We want to allow some descendants of $B$ to be used as part of the treatment-inducing proxy set $Z$ for setting up the bridge equation. These descendants, however, cannot be included in the output kernel, as their causal relation to $B$ will, in general, be confounded by $U^*$. This reflects the fact that the proximal backdoor formula in \cref{prop:proximal_backdoor_extended} requires $Z$ to be a pre-treatment proxy variable in \cref{ass:exclusion_restriction_Z_X}. In the general setting, we define $Z^* :=  \de[\mathcal{G}(R_1 \setminus W,S)](B) \cap Z$ as the subset of the treatment-inducing proxies that will be causally affected in $\mathcal{G}(R_1 \setminus W,S)$ by an additional intervention on $B$.\footnote{We choose to define \texttt{Obf} using $Z^*$ because the part of the kernel that has to be removed from the output kernel is central to our discussion of kernel operations.} Finally, we define $X := R_1 \setminus (O \cup \{B\}\cup Z \cup W)$. Using the following assumptions, we will prove identifiability of 
\begin{align*}
    \dodistr{O \cup (Z \setminus Z^*) \cup X}{}{S \cup \{B\}} = \dodistr{R_1 \setminus (W \cup Z^* \cup \{B\})}{}{S \cup \{B\}}.
\end{align*}
We note that these assumptions follow from the slightly different set of assumptions in the corresponding ``Proximal Causal Inference Step'' of \citet{Shpitser2021}. We choose to replicate our assumptions from \cref{prop:Proximal_backdoor_formula} and \cref{prop:proximal_backdoor_extended} with the treatment variable $B$ and the outcome variable $O$. 

\begin{appassumption}[Latent ignorability]
\label{app:ignorability_given_Z}
$O(s,b) \indep B(s) \, | \, U^*(s), (Z \setminus Z^*)(s),X(s)$ for all $s \in \mathcal{X}_S, b \in \mathcal{X}_B$.    
\end{appassumption}
\begin{appassumption}[Valid outcome-inducing proxy]
\label{app:kernel_valid_outcome_inducing_proxy}
$W(s) \indep (Z(s),B(s)) \, | \, U^*(s),X(s)$ for all $s \in \mathcal{X}_S$.
\end{appassumption}
\begin{appassumption}[Valid treatment-inducing proxy]
\label{app:kernel_valid_treatment_inducing_proxy}
$O(s) \indep Z(s) \, | \, B(s),U^*(s),X(s)$ for all $s \in \mathcal{X}_S$.
\end{appassumption}
Using the kernel notation, the completeness \cref{ass:outcome_completeness} generalizes as follows.
\begin{appassumption}[Completeness]
\label{app:outcome_completeness}
For any $g \in L^2(U^*)$  and all $b \in \mathcal{X}_B,x \in \mathcal{X}_X,s \in \mathcal{X}_S$
\begin{align*}
    \sum_{u^*} g(u^*) \, \dodistr{u^*}{b,z,x}{s} = 0 \text{ $P_Z$-almost surely }\ \Longleftrightarrow \ g(U^*) = 0 \text{ almost surely}.
\end{align*}
\end{appassumption}

\begin{appassumption}[Outcome bridge function]
    \label{app:outcome_bridge_function}
    For all $s \in \mathcal{X}_S$ there exists an outcome bridge function $h^{(s)}(o,w,b,x)$, solving the integral equation
    \begin{align}
        \label{eq:app_outcome_bridge_function}
        \sum_{w} h^{(s)}(o,w,b,x) \, \dodistr{w}{z,b,x}{s} = \dodistr{o}{z,b,x}{s}.
    \end{align}
\end{appassumption}
By definition, we have $O\cup Z \cup X \cup \{B\}\subseteq R_1$. We note that $W \cup Z \cup X \cup \{B\}= W \cup (R_1 \setminus O) $ such that either $W \cup Z \cup X \cup \{B\} \subseteq R_1$ or $W \cup Z \cup X \cup \{B\} \subseteq R_2$ by assumption. Thus, all kernels that are required for defining the outcome bridge functions are given to the kernel operation as input. 

We now establish the identification of the outcome kernel of \texttt{Obf} by following the arguments used in the proofs of \cref{prop:Proximal_backdoor_formula} and \cref{prop:proximal_backdoor_extended}, taking particular care to distinguish between the full set of treatment-inducing proxies $Z$ and its subset $Z \setminus Z^*$.

\begin{proposition}
    \label{prop:obf_kernel_operation_valid}
    The kernel operation \texttt{Obf}, with input, output, and conditions as defined in \cref{sec:Kernel_operations}, is valid.  In particular, we require \cref{app:ignorability_given_Z,app:kernel_valid_outcome_inducing_proxy,app:kernel_valid_treatment_inducing_proxy,app:outcome_bridge_function,app:outcome_completeness}.  
    \begin{proof} 
    Following the same arguments as in the proof of \cref{prop:Proximal_backdoor_formula}, we can use \cref{app:kernel_valid_outcome_inducing_proxy,app:kernel_valid_treatment_inducing_proxy,app:outcome_completeness,app:outcome_bridge_function} and positivity to conclude 
    \begin{align}
    \label{eq:app_outcome_bridge_with_u}
    \dodistr{o}{b,u,x}{s} = \sum_{w} h^{(s)}(o,w,b,x) \, \dodistr{w}{u,x}{s}.
    \end{align}
    As all arguments until this point did not include intervention on $B$, we could use the entire set $Z$ as treatment-inducing proxies. Notably, we allowed for descendants of $B$ as part of the treatment-inducing proxies in the construction of the bridge equation. Likewise, we also required the completeness condition to include the entire set $Z$. Now we proceed as in \cref{prop:proximal_backdoor_extended} and consider the intervention on $B$. All following arguments only use the reduced set $\tilde{Z} :=Z \setminus Z^*$. Since neither $X$ nor $\tilde{Z}$ contain any descendants of $B$, we may conclude the exclusion restriction
    $X(s,b) = X(s)$ and $\tilde{Z}(s,b) = \tilde{Z}(s)$. Thus, we have by \cref{app:ignorability_given_Z} and consistency
    \begin{align*}
        \dodistr{o,\tilde{z},x}{}{s,b} = \sum_w h^{(s)}(o,w,b,x) \, \dodistr{w,\tilde{z},x}{}{s},
    \end{align*}
    using the same derivation as in \cref{prop:proximal_backdoor_extended} with \cref{app:kernel_valid_outcome_inducing_proxy,app:kernel_valid_treatment_inducing_proxy}. Notably, the identification formula needs $\dodistr{w,\tilde{z},x}{}{s}$ to be accessible from the inputs of the kernel operation. This follows from the fact that $W \cup Z \cup X \cup \{B\}$ contains $W \cup \tilde{Z} \cup X$ as a subset and, as established above, is itself entirely contained in either $R_1$ or $R_2$.
    \end{proof}
\end{proposition}

\subsection{Treatment bridge function kernel operation \texttt{Tbf}}
\label{sec:proof_Tbf_validity}
We state the additional conditions of the kernel operation $\texttt{Tbf}$ and prove the identification of the output kernel. As before, let $\dodistr{R}{}{S}$ be the input kernel, $B \in R$ the treatment variable, $U^* \subseteq U$ the latent confounders, $Z$ the treatment-inducing proxies, and $W \subseteq R$  the outcome-inducing proxy variables. Let $O := \de[\mathcal{G}(R \setminus Z,S)](B) \setminus (W \cup \{B\})$. 

Similar to \texttt{Obf}, we want to allow the set of outcome-inducing proxies $W$ to contain descendants of $B$ for setting up the bridge equation. Again, we will have to exclude these descendants from the output kernel, reflecting the fact that the proximal IPW formula in \cref{prop:proximal_IPW_formula_extended} requires $W$ to be a pre-treatment proxy variable in \cref{ass:exclusion_restriction_W_X}. To this end, we define $W^* :=  \de[\mathcal{G}(R \setminus Z,S)](B) \cap W$ as the subset of the outcome-inducing proxies that will be causally affected by an additional intervention on $B$. Finally, we define $X := R \setminus (O \cup \{B\}\cup Z \cup W)$. Using the following assumptions, we prove identifiability of 
\begin{align}
    p(O \cup (W \setminus W^*) \cup X\, || \, S \cup \{B\}) = \dodistr{R \setminus (Z \cup W^* \cup \{B\})}{}{S \cup \{B\}}.
\end{align}
Again, we replicate the assumptions from \cref{prop:Proximal_IPW_formula} and \cref{prop:proximal_IPW_formula_extended} with the treatment variable $B$ and the outcome variable $O$. Compared to $\texttt{Obf}$, we need the following additional assumptions.
\begin{appassumption}
\label{app:ignorability_given_W}
$O(s,b) \indep B(s) \, | \, U^*(s), (W \setminus W^*)(s),X(s)$ for all $s \in \mathcal{X}_S, b \in \mathcal{X}_B$.    
\end{appassumption}

\begin{appassumption}
\label{app:valid_treatment_inducing_proxy_with_W}
$O(s) \indep Z(s) \, | \, (W \setminus W^*)(s),B(s),U^*(s),X(s)$  for all $s \in \mathcal{X}_S$.
\end{appassumption}

\begin{appassumption}[Completeness]
\label{app:treatment_completeness}
For any $g \in L^2(U^*)$ and all $b \in \mathcal{X}_B,x \in \mathcal{X}_X,s \in \mathcal{X}_S$
\begin{align*}
    \sum_{u^*} g(u^*) \, \dodistr{u^*}{b,w,x}{s} = 0 \text{ $P_W$-almost surely }\ \Longleftrightarrow \ g(U^*) = 0 \text{ almost surely}.
\end{align*}
\end{appassumption}

\begin{appassumption}[Treatment bridge function]
    \label{app:treatment_bridge_function}
    For all $s \in \mathcal{X}_S$ there exists a treatment bridge function $q^{(s)}(z,b,x)$, solving the integral equation
    \begin{align}
        \label{eq:app_treatment_bridge_function}
        \sum_{z} q^{(s)}(z,b,x) \, \dodistr{z}{w,b,x}{s} = \frac{1}{\dodistr{b}{w,x}{s}},
    \end{align}
    assuming additionally that $\dodistr{b}{w,x}{s} > 0$ for all $b \in \mathcal{X}_B, w \in \mathcal{X}_W$, $x\in \mathcal{X}_X$ and $s \in \mathcal{X}_S$.
\end{appassumption}

Now we can prove the identification of the outcome kernel of \texttt{Tbf} by following the proof of \cref{prop:Proximal_IPW_formula} and \cref{prop:proximal_IPW_formula_extended}, again taking particular care to distinguish between the full set of outcome-inducing proxies $W$ and its subset $W \setminus W^*$.

\begin{proposition}
    \label{prop:tbf_kernel_operation_valid}
    The kernel operation \texttt{Tbf}, with input, output, and conditions as defined in \cref{sec:Kernel_operations}, is valid.  In particular, we require \cref{app:ignorability_given_W,app:kernel_valid_outcome_inducing_proxy,app:treatment_completeness,app:treatment_bridge_function,app:valid_treatment_inducing_proxy_with_W}.
    \begin{proof} 
    Following the same arguments as in the proof of \cref{prop:Proximal_IPW_formula}, we can use the  \cref{app:kernel_valid_outcome_inducing_proxy,app:treatment_bridge_function,app:treatment_completeness} to conclude
    \begin{align*}
        \frac{1}{\dodistr{b}{u,x}{s}} = \sum_z q^{(s)}(z,b,x) \, \dodistr{z}{b,u,x}{s}.
    \end{align*}
    As for \texttt{Obf}, the final argument only uses the reduced set $\tilde{W} := W \setminus W^*$. 
    Since neither $X$ nor $\tilde{W}$ contain any descendants of $B$, we may conclude the exclusion restriction
    $X(s,b) = X(s)$ and $\tilde{W}(s,b) = \tilde{W}(s)$. Then we proceed as in \cref{prop:proximal_IPW_formula_extended}, using \cref{app:ignorability_given_W,app:kernel_valid_outcome_inducing_proxy,app:valid_treatment_inducing_proxy_with_W} and consistency to conclude
    \begin{align*}
        \dodistr{o,\tilde{w},x}{}{b,s} = \sum_z q^{(s)}(z,b,x) \, \dodistr{o,\tilde{w},z,b,x}{}{s}.
    \end{align*}
    As we only use a single input kernel, all kernels required to define the bridge equation and identification formula are trivially available to the kernel operation.
\end{proof}
\end{proposition}

\subsection{Extended bridge function kernel operation \texttt{Ebf}}
\label{sec:proof_Ebf_validity}
We state the additional conditions of the kernel operation $\texttt{Ebf}$ and prove the identification of the output kernel. We will not differentiate between extended outcome bridge functions and extended treatment bridge functions, as their input and output kernels are identical. We use the same notation as in \cref{sec:proof_Obf_validity} and \cref{sec:proof_Tbf_validity}.

As before, we want to allow the set of proxies $W$ and $Z$ to contain descendants of $B$ for setting up the bridge equation and define $W^* :=  \de[\mathcal{G}(R,S)](B) \cap W$ and $Z^* :=  \de[\mathcal{G}(R,S)](B) \cap Z$. Using the following assumptions, we prove identifiability of 
\begin{align}
    p(O \cup (W \setminus W^*) \cup (Z \setminus Z^*)\cup X\, || \, S \cup \{B\}) = \dodistr{R \setminus (Z^* \cup W^* \cup \{B\})}{}{S \cup \{B\}}.
\end{align}
Again, we replicate the assumptions from \cref{cor:identify_entire_joint_with_outcome_bridge} and \cref{cor:identify_entire_joint_with_treatment_bridge} with treatment variable $B$ and outcome variable $O$. Compared to $\texttt{Obf}$ and $\texttt{Tbf}$, we need the following additional assumptions.
\begin{appassumption}
\label{app:ignorability_given_W_Z}
$O(s,b) \indep B(s) \, | \, U^*(s), (W \setminus W^*)(s),(Z \setminus Z^*)(s),X(s)$ for all $s \in \mathcal{X}_S, b \in \mathcal{X}_B$.    
\end{appassumption}

\begin{appassumption}[Extended outcome bridge function]
    \label{app:extended_outcome_bridge_function}
    For all $s \in \mathcal{X}_S$ there exists an extended outcome bridge function $h^{(s)}(o,w,w',b,x)$, solving the integral equation
    \begin{align}
        \label{eq:app_kernel_extended_outcome_bridge_function}
        \sum_{w'} h^{(s)}(o,w,w',b,x) \, \dodistr{w'}{z,b,x}{s} = \dodistr{o,w}{z,b,x}{s}.
    \end{align}
\end{appassumption}

\begin{appassumption}[Extended treatment bridge function]
    \label{app:extended_treatment_bridge_function}
    For all $s \in \mathcal{X}_S$ there exists an extended treatment bridge function $q^{(s)}(z,z',b,x)$, solving the integral equation
    \begin{align}
        \label{eq:app_kernel_extended_treatment_bridge_function}
        \sum_{z'} q^{(s)}(z,z',b,x) \, \dodistr{z'}{w,b,x}{s} &= \frac{\dodistr{z}{w,x}{s}}{\dodistr{b}{w,x}{s}}.
    \end{align}
    assuming additionally that $\dodistr{b}{w,x}{s} > 0$ for all $b \in \mathcal{X}_B, w \in \mathcal{X}_W, x \in \mathcal{X}_X$ and $s \in \mathcal{X}_S$.
\end{appassumption}
\begin{proposition}
    \label{prop:ebf_kernel_operation_valid}
    The kernel operation \texttt{Ebf} with input, output, and conditions as defined in \cref{sec:Kernel_operations} is valid.  In particular, we require the \cref{app:ignorability_given_W_Z,app:kernel_valid_outcome_inducing_proxy,app:valid_treatment_inducing_proxy_with_W}. Furthermore, we need to assume the existence of either an extended outcome bridge function (\cref{app:extended_outcome_bridge_function}) or an extended treatment bridge function (\cref{app:extended_treatment_bridge_function}) and the corresponding completeness assumptions \ref{app:outcome_completeness} or \ref{app:treatment_completeness}.
    \begin{proof} 
    We only provide a proof sketch, as the derivation is very similar to the proofs of previous results. We can adapt the derivations for extended outcome bridge functions in \cref{app:identify_entire_joint_with_outcome_bridge} and for extended treatment bridge functions in \cref{app:identify_entire_joint_with_treatment_bridge}, slightly adjusting the notation to the kernel setting. As before, we use the full sets $W$ and $Z$ to formulate the bridge equations and use the reduced sets $\tilde{W} := W \setminus W^*$ and $\tilde{Z} := Z \setminus Z^*$ only for the following steps of the derivation that require the exclusion restriction $X(s,b) = X(s)$, $\tilde{W}(s,b) = \tilde{W}(s)$, and $\tilde{Z}(s,b) = \tilde{Z}(s)$. 
\end{proof}
\end{proposition}

\subsection[Proof of Theorem~\ref{thm:generalized_proximal_ID_algorithm}]{Proof of \cref{thm:generalized_proximal_ID_algorithm}}
\label{sec:proof_of_generalized_ID_algorithm}
First, we can apply the district factorization result from \cref{lem:district_factorization} with the latent projection set $H$, yielding the ADMG $\mathcal{G}(V^*)$. We consider the districts of the induced subgraph $\mathcal{G}(V^*)_{Y^*}$ of the ADMG $\mathcal{G}(V^*)$. \cref{lem:district_factorization} shows that the identification of each target kernel $\dodistr{D}{}{V^* \setminus D}$ for $D \in \mathcal{D}(\mathcal{G}(V^*)_{Y^*})$ is a sufficient condition for the identifiability of $\dodistr{Y}{}{a}$. We apply the proximal identification step described in \cref{alg:proximal_ID_step_1} to prove the identification of each target kernel. The validity of the proximal identification step follows from the validity of the kernel operations \texttt{Fix}, \texttt{Obf}, \texttt{Tbf}, \texttt{Ebf} , and \texttt{Cut} that we proved in \cref{prop:fix_kernel_operation_valid,prop:obf_kernel_operation_valid,prop:tbf_kernel_operation_valid,prop:ebf_kernel_operation_valid}. The conclusion follows from inductively applying \cref{alg:proximal_ID_step_1}. Notably, we check at each step whether the relevant variables $V^*$ are still contained in $\texttt{Var}(P_1)$. If variables in $V^*$ were removed by a kernel operation, or the chosen kernel operation is not applicable to $P_1$, the algorithm fails.

\newpage
\bibliographystyle{abbrvnat} 
\bibliography{PCI_Literature}{}

\end{document}